\documentclass[12pt]{article}
\usepackage{hyperref}
\usepackage{float}
\usepackage{amsmath}
\usepackage{graphicx,psfrag,epsf}
\usepackage{enumerate}
\usepackage{natbib}
\usepackage{multirow}
\usepackage[ruled,vlined]{algorithm2e}
\usepackage{array}

\usepackage{xcolor}

\usepackage{graphicx}
\usepackage[utf8 ]{inputenc}
\usepackage[english]{babel}
\usepackage{amsthm}
\usepackage[bottom]{footmisc}
\usepackage{longtable}
\usepackage{mathrsfs}
\usepackage{extarrows}
\usepackage{bigints}
\usepackage{subcaption}

\usepackage{setspace}

\usepackage[top=0.58in, bottom=0.68in, left=0.50in, right=0.50in]{geometry}
\usepackage{amsmath}
\usepackage{amsfonts}

\newtheorem{theorem}{Theorem}[section]
\newtheorem{lemma}[theorem]{Lemma}
\newtheorem{proposition}[theorem]{Proposition}

\theoremstyle{definition}
\newtheorem{remark}{Remark}[section]

\newtheorem{definition}{Definition}[section]

\DeclareTextFontCommand{\textmyfont}{\footnotesize}

\begin{document}
\def\spacingset#1{\renewcommand{\baselinestretch}%
{#1}\small\normalsize}

\spacingset{1.4}

\title{\bf Student-t Stochastic Volatility Model With Composite Likelihood EM-Algorithm}

\author{Raanju R. Sundararajan$^\star$\footnote{Email: \href{mailto:rsundararajan@smu.edu}{rsundararajan@smu.edu}}\,\,\, and \, Wagner Barreto-Souza$^\ddag$\footnote{Email: \href{mailto:wagner.barretosouza@kaust.edu.sa}{wagner.barretosouza@kaust.edu.sa} (Corresponding Author)} \\
	$^\star${\it \small Department of Statistical Science, Southern Methodist University,  Dallas, Texas, USA} \\
	$^\ddag${\it \small Statistics Program, King Abdullah University of Science and Technology, Thuwal, Saudi Arabia}}

\vspace{0.5cm}

\date{}

\maketitle

\begin{abstract}
A new robust stochastic volatility (SV) model having Student-$t$ marginals is proposed. Our process is defined through a linear normal regression model driven by a latent gamma process that controls temporal dependence. This gamma process is strategically chosen to enable us to find an explicit expression for the pairwise joint density function of the Student-$t$ response process. With this at hand, we propose a composite likelihood (CL) based inference for our model, which can be straightforwardly implemented with a low computational cost. This is a remarkable feature of our Student-$t$ SV process over existing SV models in the literature that involve computationally heavy algorithms for estimating parameters. Aiming at a precise estimation of the parameters related to the latent process, we propose a CL Expectation-Maximization algorithm and discuss a bootstrap approach to obtain standard errors. The finite-sample performance of our composite likelihood methods is assessed through Monte Carlo simulations. The methodology is motivated by an empirical application in the financial market. We analyze the relationship, across multiple time periods, between various US sector Exchange-Traded Funds returns and individual companies' stock price returns based on our novel Student-$t$ model. This relationship is further utilized in selecting optimal financial portfolios. \\

\noindent {\bf 2020 AMS Subject Classification}. Primary: 62M10.\\

\noindent {\bf Keywords and phrases}: Composite likelihood estimation; EM-algorithm; Exchange-Traded Funds; Quantitative trading strategies; Stochastic volatility models; Robustness; Portfolio optimization.\\

\end{abstract}

\section{Introduction} \label{s:intro}

In a pioneering paper, \cite{eng1982} introduced the class of Autoregressive Conditional Heteroscedastic (ARCH) models, which was extended by the GARCH models proposed by \cite{bol1986}. The literature on GARCH models is very rich and quite extensive. We recommend the book by \cite{frazak2019} for an updated account on this topic.

The stochastic volatility (SV) models are an important alternative to the GARCH approach \citep{brorui2004}. A standard SV model is defined by
\begin{eqnarray}\label{standSV1}
Y_t|h_t\sim N(0,\sigma^2 e^{h_t}), \quad h_t=\rho h_{t-1}+\eta_t, \quad t\in\mathbb N,
\end{eqnarray}
where $\{\eta_t\}_{t\in\mathbb N}\stackrel{iid}{\sim}N(0,\sigma^2_\eta)$, $|\rho|<1$, and $\sigma^2,\sigma_\eta^2>0$; for instance, see \cite{tay1986}, \cite{tay1994}, and \cite{kimetal1998}.

There are many variants of the SV model in (\ref{standSV1}). Aiming at robustness against the normal assumption in (\ref{standSV1}), \cite{jacetal1994}, \cite{haretal1994}, \cite{sankoo1998}, and \cite{chietal2002} considered a Student-$t$ distribution as a marginal or conditional distribution in (\ref{standSV1}). Other relevant variants are due to \cite{bn1997}, \cite{mahsch1998}, \cite{jacetal2004}, \cite{abretal2006}, \cite{nakomo2009}, \cite{jenmah2010}, and \cite{delgri2011}, just to name a few.

Markov Chain Monte Carlo strategies for estimating SV models have been extensively considered in the literature by \cite{jacetal1994}, \cite{shepit1997}, \cite{kimetal1998}, and \cite{chietal2002}. Other Bayesian inference approaches for SV models are addressed by \cite{rueetal2009}, \cite{gircal2011}, \cite{kassch2014}, \cite{tranetal2017}, \cite{viretal2019}, and \cite{gonsto2021}, among others.  
SV models under a state-space perspective were proposed by \cite{koosch2013}, \cite{gouyan2016}, \cite{cre2017}, and \cite{boretal2020}.
Quasi-likelihood and adaptive sequential Monte Carlo approaches for estimating parameters of L\'evy-driven stochastic volatility models were proposed respectively by \cite{ver2011} and \cite{jasetal2011}.

In this paper, we propose a new robust stochastic volatility model having Student-$t$ marginals. Our response process is defined by a linear normal regression model driven by a latent gamma process responsible for controlling the temporal dependence. We consider the first-order gamma autoregressive model by \cite{sim1990} as the latent process, and this is strategically chosen because it enables us to find an explicit expression for the pairwise joint density function of our Student-$t$ response process. With this setup, we propose a composite likelihood (CL) based inference for our model, which can be straightforwardly implemented with a low computational cost. This is an advantageous feature of our Student-$t$ SV model over the existing SV models in the literature that often require very time-consuming algorithms for estimating parameters. With the aim of precise estimation of the parameters related to the latent process, we also propose a CL Expectation-Maximization (CLEM) algorithm. It is worth mentioning that the optimization required in the proposed CLEM algorithm is much simpler than the direct maximization of the log-composite likelihood function for our model since it depends on the hypergeometric function. This will be discussed in detail in Section \ref{sec:pairwise_likelihood}.

Only a few papers have addressed estimation for stochastic volatility models based on a composite likelihood approach. For instance, the methods in \cite{barcho2014} and \cite{chanetal2018} use a Bayesian approach, and \cite{gouyan2016} use a frequentist perspective, with the latter being related to our work. There, the authors propose an SV model with Student-$t$ marginals with an estimation based on a composite likelihood approach. We call the attention that the latent gamma process we are using here is different from that in \cite{gouyan2016} and this choice leads us to obtain an explicit and simple form for the pairwise likelihood function of the response process. Further, as a consequence, a composite likelihood EM-algorithm with a closed-form for the E-step and simple maximization for the M-step is proposed here; this is not considered in \cite{gouyan2016}. Another critical difference is that our model allows for covariates, which is crucial for our empirical application to financial portfolio management.

We now summarize the main features and contributions of our paper as follows. 
\begin{itemize}

	\item As argued by \cite{sankoo1998} (see Subsection 5.2 of their work), larger kurtosis is often observed in the marginal density of financial data and a Student-$t$ distribution is recommended for this task. By construction, our model includes this feature by resulting in a Student-$t$ distribution as the marginal distribution of the response process.
	
	\item  Our inferential procedure is robust to misspecification since it does not require a full-likelihood specification, but just the pairwise density functions. Furthermore, our model is robust as it is equipped to deal with outliers due to the Student-$t$ marginals.

	\item We propose a composite likelihood EM-algorithm method with a simple optimization procedure that efficiently estimates all the model parameters, including those related to the latent process, which is a well-known challenging problem in practice. The computer code for implementing our method in \texttt{R} is made available. 
	
	\item The composite likelihood-based inferential procedures discussed in this paper demand low to moderate computational cost relative to the estimation techniques for SV models discussed in existing literature, especially those involving Bayesian and/or Monte Carlo strategies, where complex and time-consuming algorithms are required.
	
	\item In our financial data application, we apply our method to find relationships, across multiple time periods, between US sector ETFs and individual companies' stock price returns and use them for optimal portfolio selection.  Existing methods in the financial literature use penalized regressions, like LASSO, and do not provide many details on statistical inference and are also less suitable in a factor regression setup with dependent factors. Our formulation, resulting in Student-$t$ marginals of the response, is more appropriate for financial data and provides means for carrying out inference which is critical in understanding which sector ETFs are important in explaining movements in stock price returns of individual companies.    
\end{itemize}

The paper is organized in the following manner. In Section \ref{sec:model}, we define our new Student-$t$ stochastic volatility model and obtain some of its properties. In particular, we provide an explicit expression for the joint density function of any pair of observations from the process, which plays a central role in this paper. Statistical inference based on the CL approach is addressed in Section \ref{sec:pairwise_likelihood}. We also develop a CLEM algorithm and discuss how to obtain standard errors using a bootstrap approach. Monte Carlo simulations to study and compare the CL and CLEM approaches are presented in Section \ref{sec:simulation}. In Section \ref{sec:application}, we apply the proposed Student-$t$ SV model to analyze the relationship, across multiple time periods, between various US sector Exchange-Traded Funds returns and individual companies' stock price returns. This relationship is further utilized in selecting optimal financial portfolios. Concluding remarks and topics for future research are discussed in Section \ref{sec:conclusion}. This paper contains Supplementary Material, which can be obtained from the authors upon request.

\section{A novel Student-$t$ stochastic volatility model}\label{sec:model}

In this section, we introduce our proposed Student-$t$ stochastic volatility model and provide some of its properties. We begin by introducing some needed notation.

Throughout this paper, we denote a random variable $Z$ following a gamma distribution with moment generating function $\Psi_Z(s)=E(\exp\{sZ\})=(1-s/\gamma)^{-\nu}$, for $s<\gamma$, by $Z\sim\mbox{G}(\gamma,\nu)$, where $\gamma>0$ and $\nu>0$ denote respectively the rate and shape parameters; the exponential case (with $\nu$=1) is denoted by $Z\sim\mbox{Exp}(\theta)$. Further, we denote a Poisson random variable $N$ with mean $\lambda>0$ by $N\sim\mbox{Pois}(\lambda)$. 

We now present the gamma autoregressive (GAR) process by \cite{sim1990}. We fix the mean of the process equal to 1 to avoid non-identifiability problems in what follows. 

\begin{definition} \label{def:gar}
	The first-order $\mbox{GAR}$ process (with mean 1) is defined by the stochastic process $\{Z_t\}_{t\in\mathbb N}$ satisfying the stochastic equation	$Z_t=\alpha\odot Z_{t-1}+\epsilon_t$,  $t\in\mathbb N$, $Z_0\sim \mbox{G}(\phi,\phi)$,	where the operator $\odot$ is defined by $\alpha\odot Z_{t-1}\stackrel{d}{=}\sum_{i=1}^{N_{t-1}}W_i$, with $N_{t-1}|Z_{t-1}=z\sim\mbox{Pois}(\alpha\rho z)$, $\{W_i\}_{i=1}^\infty\stackrel{iid}{\sim}\mbox{Exp}(\alpha)$ and  $\{\epsilon_i\}_{i=1}^\infty\stackrel{iid}{\sim}\mbox{G}(\alpha,\phi)$ are assumed to be independent and $\alpha=\dfrac{\phi}{1-\rho}$, for $\phi>0$ and $\rho\in(0,1)$.
\end{definition}

In the following proposition, we state some important results of the GAR(1) process due to \cite{sim1990}, except for the strongly mixing property which was established by \cite{bso2021}.

\begin{proposition} Let $\{Z_t\}_{t\in\mathbb N}\sim\mbox{GAR}(1)$. We have that  $\{Z_t\}_{t\in\mathbb N}$ is a strongly mixing and a stationary gamma process with marginal mean and variance 1 and $\phi^{-1}$, respectively. Further, we have that: 
	
	\noindent (i) The joint density function of $(Z_{t+j},Z_t)$ is given by
	\begin{eqnarray}\label{eq:jointZ}
	g(z_{t+j},z_t)=\dfrac{\phi^{\phi+1}}{(1-\rho^j)\Gamma(\phi)}\left(\dfrac{z_{t+j}z_t}{\rho^j}\right)^{\frac{\phi-1}{2}}
	\exp\left(-\phi\dfrac{z_{t+j}+z_t}{1-\rho^j}\right)I_{\phi-1}\left(2\phi\dfrac{\sqrt{\rho^j z_{t+j}z_t}}{1-\rho^j}\right),
	\end{eqnarray}
	for $z_{t+j},z_t>0$ and $t,j\in\mathbb N$, where $I_{\nu}(x)=\displaystyle\sum_{k=0}^\infty\dfrac{(x/2)^{2k+\nu}}{\Gamma(k+\nu+1)k!}$
	is the modified Bessel function of the first kind of order $\nu\in\mathbb R$, for $x\in\mathbb R$;
	
	\noindent (ii) The autocorrelation of the gamma $\mbox{AR}(1)$ process is $\mbox{corr}(Z_{t+j},Z_t)=\rho^j$ and the conditional moments are
	$E\left(Z_{t+j}^k|Z_t=z\right)=k!\left(\dfrac{1-\rho^j}{\phi}\right)^kL_k^{\phi-1}\left(-z\dfrac{\phi\rho^j}{1-\rho^j}\right)$,
	for $k\in\mathbb N$, where $L_n^{\nu}(x)=\dfrac{x^{-\nu}\exp(x)}{n!}\dfrac{d^n}{dx^n}\big\{x^{n+\nu}\exp(-x)\big\}$ is the generalized Laguerre polynomial of degree $n$, for $n\in\mathbb N$ and $x,\nu>0$. 
\end{proposition}

\begin{remark}\label{rem_sto_rep}
	The joint density function (\ref{eq:jointZ}) corresponds to a Kibble bivariate gamma distribution \citep{kib1941}. From this fact and using the results from \cite{ilietal2005}, we obtain that $(Z_{t+j},Z_t)$ satisfies the stochastic representation 
	\begin{eqnarray}\label{sto_rep}
	Z_{t+j},Z_t|U=u\stackrel{iid}{\sim} G\left(\dfrac{\phi}{1-\rho^j},\phi+u\right),
	\end{eqnarray}
	where $U$ is a negative binomial random variable with probability function $P(U=u)=\dfrac{\Gamma(\phi+u)}{\Gamma(\phi)u!}(1-\rho^j)^\phi\rho^{ju}$, for $u=0,1,\hdots$. This result will be important in the development of a composite likelihood EM-algorithm approach in Section \ref{sec:pairwise_likelihood}.
\end{remark}	

 We are now ready to introduce our Student-$t$ process as follows.

\begin{definition} \label{def:tar}({\it New Student-$t$ SV model}) Assume $\{Z_t\}_{t\in\mathbb N}$ is a latent $\mbox{GAR}(1)$ process according Definition \ref{def:gar} with parameters $\phi=\dfrac{\nu}{2}$ and $\rho\in(0,1)$, where $\nu>0$.  We say that a process $\{Y_t\}_{t\in\mathbb N}$ is a New Stochastic Volatility model with Student-$t$ marginals (in short NSVt) if it is conditionally independent given the latent process $\{Z_t\}_{t\in\mathbb N}$ and satisfies the following stochastic representation:
	\begin{eqnarray*}
		Y_t|\{Z_j\}_{j\in\mathbb N}\sim N\left(\mu_t,Z_t^{-1}\sigma^2\right),\quad t\in\mathbb N,
	\end{eqnarray*}
	where $\mu_t={\bf x}_t^\top{\boldsymbol\beta}$, ${\bf x}_t$ is an observed $p$-dimensional covariate vector with $\boldsymbol\beta=(\beta_1,\ldots,\beta_p)^\top\in\mathbb R^p$ being the associated vector of regression coefficients, and $\sigma^2>0$. We denote $\{Y_t\}_{t\in\mathbb N}\sim\mbox{NSV}t$ which implicitly depends on the parameter vector $\boldsymbol\theta=(\boldsymbol\beta^\top,\sigma^2,\nu,\rho)^\top$.
\end{definition}

\begin{remark}
\cite{cre2017} also considered a stochastic volatility model constructed by a normal linear regression driven by the latent gamma autoregression. However, the latent process in their work enters the variance of the normal regression in a multiplicative way, thereby yielding normal gamma (also known as variance gamma) marginals. This is different from the model introduced in Definition \ref{def:tar}. In our case, the gamma autoregression enters the normal model in an ``inverse" way (inverse-gamma), which implies that our process has Student-$t$ marginals. 
\end{remark}

The parameter $\rho$ controls the dependence/persistence of volatility in our model while $\nu$ is responsible for the tail behavior. The scale of the process is controlled by the parameter $\sigma^2$. In Figure \ref{fig:trajectory1}, we plot simulated trajectories of the NSVt process $\{Y_t\}_{t\in\mathbb N}$ from Definition \ref{def:tar} for some values of the parameter $\nu$. Here we fix $n=200$ (sample size), $\sigma^2=1$, $\rho=0.8$, and the mean $\mu_t=0$ $\forall t$. As expected, we can observe more extreme observations for the cases $\nu=1$ and $\nu=2$ in comparison with the plots for $v=4$ and $\nu=6$. Additional plots for different choices of the parameter $\rho$ are given in the Supplementary Material.

\begin{figure*}
	\begin{center}
		\includegraphics[scale=0.5]{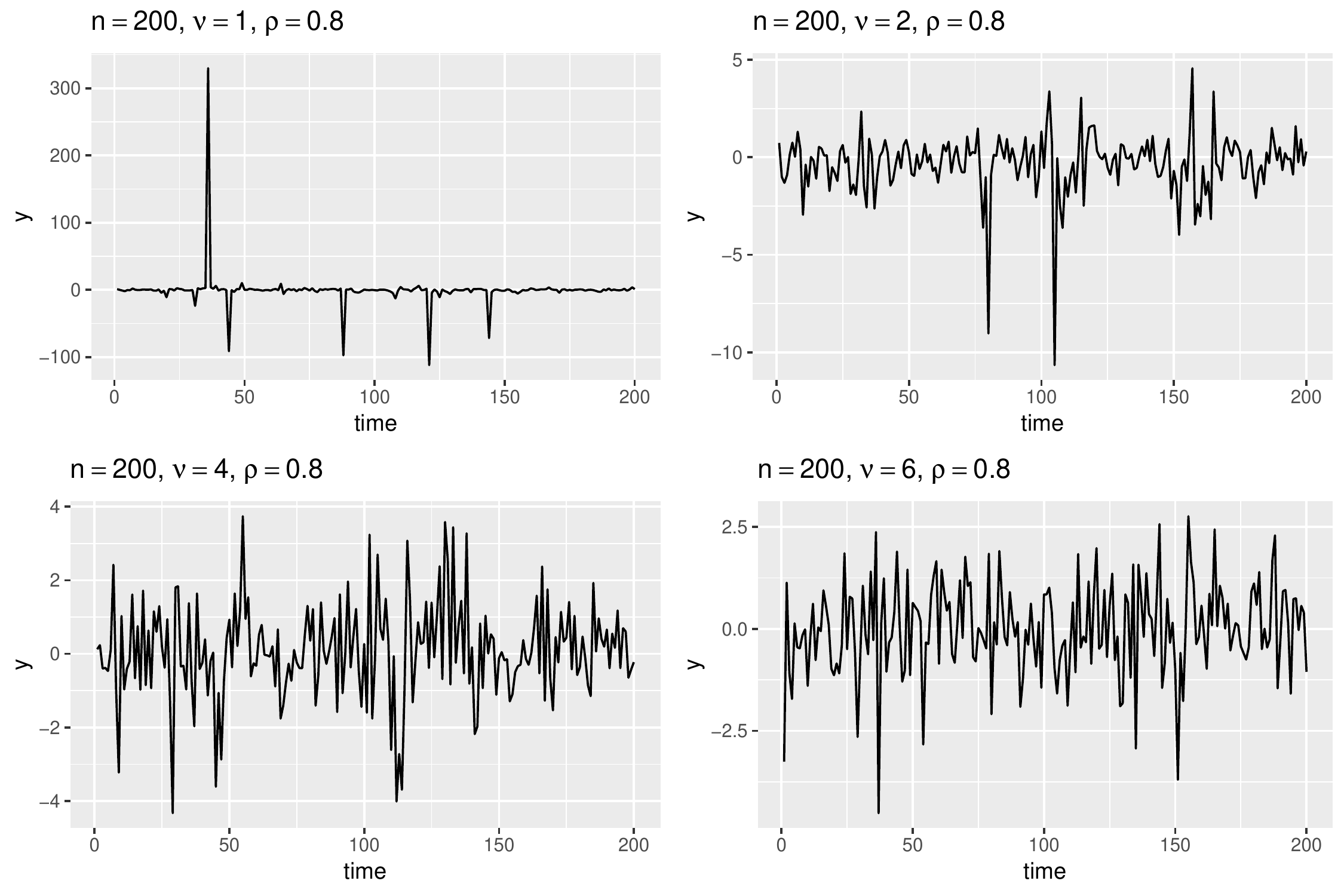}
		\caption{Plots of simulated trajectories of the NSVt process with $n=200$, $\sigma^2=1$, $\rho=0.8$, and $\mu_t=0$, for some values of $\nu$.  } \label{fig:trajectory1}
	\end{center}
\end{figure*}

\begin{remark}
	For inferential purposes, we will assume in Section \ref{sec:pairwise_likelihood} that $\{{\bf x}_t\}_{t\in\mathbb N}$ is the observed trajectory of a multivariate stochastic process $\{{\bf X}_t\}_{t\in\mathbb N}$ satisfying some desirable properties to ensure consistency and asymptotic normality of the proposed estimators as stated in the literature; see Section \ref{sec:pairwise_likelihood}. For the results in this section, this assumption is not necessary.
\end{remark}

The explicit expression for the joint density function of any pair of the NSVt process is provided in the next proposition, which plays a fundamental role in our purposes. The proofs of the results stated in this paper can be found in the Appendix. 

\begin{proposition}\label{p:jointdensity}
	Let $\{Y_t\}_{t\in\mathbb N}\sim\mbox{NSV}t$ according to Definition \ref{def:tar}. Then, for $j,t\in\mathbb N$, the joint density function of $(Y_{t+j},Y_t)$ is given by
	\begin{eqnarray*}
		f(y_{t+j},y_t)=\dfrac{(1-\rho^{j})^{\frac{\nu}{2}+1}}{\pi\nu\sigma^2}\left(\frac{\Gamma\left(\frac{\nu+1}{2}\right)}{\Gamma\left(\frac{\nu}{2}\right)}\right)^2\omega(y_{t+j},y_t)^{\frac{\nu+1}{2}}{_2}F_1\left(\dfrac{\nu+1}{2},\dfrac{\nu+1}{2};\dfrac{\nu}{2},\rho^j\omega(y_{t+j},y_t)\right),
	\end{eqnarray*}
	for $y_{t+j},y_t\in\mathbb R$, where $\omega(y_{t+j},y_t)\equiv\left(1+(1-\rho^j)\dfrac{(y_{t+j}-\mu_{t+j})^2}{\nu\sigma^2}\right)^{-1}
		\left(1+(1-\rho^j)\dfrac{(y_t-\mu_t)^2}{\nu\sigma^2}\right)^{-1}$
	and\\  ${_2F}_1(a_1,a_2;b_1,d)$ is the hypergeometric function\footnote{\url{https://mathworld.wolfram.com/GeneralizedHypergeometricFunction.html}.} defined for $a_1,a_2,b_1>0$ and $|d|<1$.
\end{proposition}

 We conclude this section by obtaining the conditional variance of $Y_{t+j}$ given $Y_t$, which can be of practical interest for forecasting volatility.

\begin{proposition}\label{p:condvar}
	Let $\{Y_t\}_{t\in\mathbb N}\sim\mbox{NSV}t$. For $\nu>2$, the conditional variance of $Y_{t+j}$ given $Y_t$ is given by 
	\begin{eqnarray}\label{eq:condvar_expression}
		\mbox{Var}(Y_{t+j}|Y_t=y_t)=\dfrac{\sigma^2\nu\left(1-w_j(y_t)\right)^{(\nu+1)/2}}{2(1-\rho^j)\Gamma\left(\nu/2-1\right)}	
		{_1}F_1\left(\dfrac{\nu+1}{2},\dfrac{\nu}{2}-1;w_j(y_t)\right),	
	\end{eqnarray}
	where $w_j(y_t)\equiv\dfrac{\rho^j}{1+(1-\rho^j)(y_t-\mu_t)^2/(\nu\sigma^2)}$ and ${_1F}_1(\cdot,\cdot;\cdot)$ is the confluent hypergeometric function of the first kind, also known as Kummer's function\footnote{ \url{https://mathworld.wolfram.com/ConfluentHypergeometricFunctionoftheFirstKind.html}}.
\end{proposition}

\section{Composite likelihood inference} \label{sec:pairwise_likelihood}

Here we describe the estimation of parameters in our proposed Student-$t$ stochastic volatility model via composite likelihood approach \citep{varetal11}. Let $y_1,\ldots,y_n$ be an observed trajectory of the NSVt process. With the aid of an explicit form for the joint density function derived in Proposition \ref{p:jointdensity}, we are able to perform composite/pairwise likelihood inference for our proposed Student-$t$ SV model. The pairwise likelihood function of order $m\in\mathbb N$ is given by   
\begin{eqnarray} \label{eq:cl_likelihood}
\mbox{PL}^{(m)}(\boldsymbol\theta)=\prod_{t=m+1}^n\prod_{i=1}^m f(y_t,y_{t-i}),
\end{eqnarray}
where $f(\cdot,\cdot)$ is the joint density function in Proposition \ref{p:jointdensity}. Hence, the log-pairwise likelihood function $\ell^{(m)}(\boldsymbol\theta)\equiv \log\mbox{PL}^{(m)}(\boldsymbol\theta)$ assumes the form
{\small \begin{eqnarray} \label{eq:cl_log_likelihood}
&&\ell^{(m)}(\boldsymbol\theta)\propto  (n-m)m\bigg\{2\left(\log\Gamma\left(\dfrac{\nu+1}{2}\right)-\log\Gamma\left(\dfrac{\nu}{2}\right)\right) -\log\nu-\log\sigma^2\bigg\}+ (n-m)\left(\dfrac{\nu}{2}+1\right)\times \notag \\
&&\sum_{i=1}^m\log(1-\rho^i) +\sum_{t=m+1}^n\sum_{i=1}^m\bigg\{\dfrac{\nu+1}{2} \log \omega(y_t,y_{t-i})+ \log {_2}F_1\left(\dfrac{\nu+1}{2},\dfrac{\nu+1}{2};\dfrac{\nu}{2},\rho^i\omega(y_t,y_{t-i})\right)\bigg\}.
\end{eqnarray}}

The pairwise likelihood estimator is obtained by $\widehat{\boldsymbol\theta}=\mbox{argmax}_{\boldsymbol\theta}\,\ell^{(m)}(\boldsymbol\theta)$. To do this in practice, we need some optimization procedures such as Newton-Raphson or BFGS methods. In this paper, we use the \texttt{optim} package from the \texttt{R} program for this purpose. The hypergeometric function in the above log-pairwise likelihood function is implemented through the \texttt{hypergeo} package in \texttt{R}.

Initial guesses for the above optimization problem can be obtained through a Student-$t$ regression model. More precisely, since the marginals of our response process are Student-$t$ distributed, we can consider the following pseudo log-likelihood function to obtain initial guesses for $\boldsymbol\beta$, $\sigma^2$ and $\nu$:
\begin{eqnarray} \label{eq:student_t_initial_guess_lkl}
\widetilde\ell(\boldsymbol\beta,\sigma^2,\nu)\propto n\bigg\{ \log\Gamma\left(\dfrac{\nu+1}{2}\right)-\log\Gamma\left(\dfrac{\nu}{2}\right)-  \dfrac{1}{2}\log\sigma^2-\dfrac{1}{2}\log\nu\bigg\}-\dfrac{\nu+1}{2}\sum_{t=1}^n\log\left(1+\dfrac{(y_t-\mu_t)^2}{\nu\sigma^2}\right).
\end{eqnarray}
For the parameter $\rho$, we recommend using a grid of values in $(0,1)$ as possible initial guesses.

The standard errors of the composite likelihood estimates can be obtained through the \cite{godambe1960} information given by ${\bf H}(\boldsymbol\theta)^\top{\bf J}(\boldsymbol\theta)^{-1}{\bf H}(\boldsymbol\theta)$, where ${\bf H}(\boldsymbol\theta)=-E\left(\dfrac{\partial^2\ell^{(m)}(\boldsymbol\theta)}{\partial\boldsymbol\theta\partial\boldsymbol\theta^\top}\right)$ and ${\bf J}(\boldsymbol\theta)=\mbox{var}\left(\dfrac{\partial\ell^{(m)}(\boldsymbol\theta)}{\partial\boldsymbol\theta}\right)$. These matrices are hard to compute and are further cumbersome since they involve the first and second order derivatives of the hypergeometric function. Therefore, we recommend a parametric bootstrap to obtain standard errors and confidence intervals of the parameter estimates. This is further explored in finite sample situations in Section \ref{sec:application}.  

One remarkable feature of using the composite likelihood estimation based on (\ref{eq:cl_log_likelihood}) is the low computational cost. This is very important for practitioners dealing with high-frequency financial data where the ability to make shorter horizon portfolio decisions can be extremely useful. For instance, \cite{liu2009} and \cite{haut2015} discuss the importance of high-frequency financial data in assisting managers to make portfolio decisions over a shorter horizon, such as over a few days, as opposed to making similar decisions over longer periods. On the other hand, some numerical issues might be experienced in the optimization procedure of the direct maximization of the composite log-likelihood function. To overcome such a problem, we now propose a composite likelihood EM-algorithm approach.

\subsection{Composite likelihood EM-algorithm} \label{sec:clem_section}

The maximization of the log-pairwise likelihood function in \eqref{eq:cl_log_likelihood} might pose computational difficulties due to the hypergeometric function. As an alternative, we propose a composite likelihood EM (CLEM) algorithm \citep{gaoson2011} for estimating model parameters, which is one of the variants of the classical EM-algorithm by \cite{demetal1977}. We start by discussing the E-step associated with our CLEM technique. Here, we will assume that $\nu$ is known. Otherwise, we will not be able to provide a closed-form E-step for our algorithm. In Remark \ref{nu_known}, we provide some approaches for dealing with this assumption in practice. 

\subsubsection*{CL E-step}

Without loss of generality, consider the pair $(Y_t,Y_{t-i})$. A natural choice for the latent effects is the pair $(Z_t,Z_{t-i})$, but its joint density function involves a Bessel function, which makes it hard to obtain an explicit form for the E-step of the algorithm, and further complicates the maximization of the corresponding function in the M-step. A key idea to overcome this problem is to use the stochastic representation given in (\ref{sto_rep}). Besides, we consider the negative binomial distribution in Remark \ref{rem_sto_rep}, represented as a mixed Poisson law.

The observable time series pair is ${\bf Y}^{obs}_{ti}=(Y_t,Y_{t-i})$ and the complete data is ${\bf Y}^{c}_{ti}=(Y_t,Y_{t-i},Z_{t1},Z_{t-i\,1},Z_{t2},\\Z_{t-i\,2},U_{ti})$, where $Z_t=Z_{t1}+Z_{t2}$, $Z_{t-i}=Z_{t-i\,1}+Z_{t-i\,2}$, with $Z_{t1},Z_{t2},Z_{t-i\,1},Z_{t-i\,2}$ all conditionally independent given $U_{ti}$. We have that the following hierarchical representation holds:
\begin{eqnarray*}\label{hierarchical}
&&\begin{pmatrix}
	Y_t\\
	Y_{t-i}
\end{pmatrix}\bigg| \begin{pmatrix}
Z_t\\
Z_{t-i}
\end{pmatrix}\sim N_2
\begin{pmatrix}
	\begin{pmatrix}
		\mu_t\\
		\mu_{t-i}
	\end{pmatrix}\!\!,&
	\sigma^2\begin{pmatrix}
		Z_t^{-1} & 0 \\
		0 & Z_{t-i}^{-1}
	\end{pmatrix}
\end{pmatrix},\nonumber\\
&&(Z_{t1},Z_{t-i\,1})|U_{ti}\stackrel{iid}{\sim} G\left(\dfrac{\nu}{2(1-\rho^i)},\dfrac{\nu}{2}\right),\nonumber\\
&&(Z_{t2},Z_{t-i\,2})|U_{ti}\stackrel{iid}{\sim} G\left(\dfrac{\nu}{2(1-\rho^i)},U_{ti}\right),\nonumber\\
&&U_{ti}|\Omega_{ti}\sim \mbox{Pois}\left(\dfrac{\rho^i}{1-\rho^i}\Omega_{ti}\right),\,
\Omega_{ti}\sim G\left(1,\dfrac{\nu}{2}\right).
\end{eqnarray*}

The CL-E step for a single ${\bf Y}^{c}_{ti}$ is given by the $Q_{ti}$-function: 
\begin{eqnarray*}
Q_{ti}(\boldsymbol\theta;\boldsymbol\theta^{(r)})=\int_{\mathcal S_{{\bf y}^{lat}_{ti}}} \log f({\bf y}^{c}_{ti};\boldsymbol\theta) dF({\bf y}^{c}_{ti}|{\bf y}^{obs}_{ti};\boldsymbol\theta^{(r)}),
\end{eqnarray*}	
where $\boldsymbol\theta^{(r)}$ denotes the estimate of $\boldsymbol\theta$ in the $r$th loop of the CLEM algorithm, $\mathcal S_{{\bf y}^{lat}_{ti}}$ is the support of the latent random vector ${\bf Y}^{lat}_{ti}=(Z_{t1},Z_{t-i\,1},Z_{t2},Z_{t-i\,2},U_{ti})$, $f({\bf y}^{c}_{ti};\boldsymbol\theta)$ is the joint density function of ${\bf Y}^{obs}_{ti}$, $F({\bf y}^{c}_{ti}|{\bf y}^{obs}_{ti};\boldsymbol\theta^{(r)})$ is the conditional distribution function of ${\bf Y}^{c}_{ti}$ given ${\bf Y}^{obs}_{ti}={\bf y}^{obs}_{ti}$ evaluated at $\boldsymbol\theta^{(r)}$, for $t=m+1,\ldots,n$, and $i=1,\ldots,m$.

The $Q_{ti}$-function can be decomposed as $Q_{ti}(\boldsymbol\theta;\boldsymbol\theta^{(r)})=Q^{(1)}_{ti}(\boldsymbol \beta,\sigma^2;\boldsymbol\theta^{(r)})+Q^{(2)}_{ti}(\nu,\rho;\boldsymbol\theta^{(r)})$, with
\begin{eqnarray*}
Q^{(1)}_{ti}(\boldsymbol \beta,\sigma^2;\boldsymbol\theta^{(r)})&\propto&
 -\log\sigma^2-\dfrac{1}{2\sigma^2}\bigg\{ \zeta_{ti}^{(r)}(Y_t-\mu_t)^2+\zeta_{t-i\, i}^{(r)}(Y_{t-i}-\mu_{t-i})^2\bigg\},\\
Q^{(2)}_{ti}(\nu,\rho;\boldsymbol\theta^{(r)})&\propto&
		-(2\tau^{(r)}_{ti}+\nu/2)\log\left(1-\rho^i\right)-\dfrac{\nu(\zeta^{(r)}_{ti\,1}+\zeta^{(r)}_{ti\,2})}{2(1-\rho^i)}+\tau_{ti}^{(r)}\log \rho^i,
\end{eqnarray*}
where $\zeta^{(r)}_{ti\,1}=E(Z_t|{\bf Y}^{obs}_{ti};\boldsymbol\theta^{(r)})$, $\zeta^{(r)}_{ti\,2}=E(Z_{t-i}|{\bf Y}^{obs}_{ti};\boldsymbol\theta^{(r)})$, and $\tau_{ti}^{(r)}=E(U_{ti}|{\bf Y}^{obs}_{ti};\boldsymbol\theta^{(r)})$ are the conditional expectations required for our proposed CLEM algorithm. Closed-forms for these conditional expectations are provided in the next proposition; its proof is given in the Appendix. For simplicity of notation, the superscript $(r)$ is omitted.

\begin{proposition}\label{E-step}
	The conditional expectations involved in our proposed CLEM-algorithm are given by
	\begin{eqnarray*}
		\zeta_{ti\,1}=\dfrac{2(1-\rho^i)^{\frac{\nu}{2}+2}\Gamma\left(\dfrac{\nu+1}{2}\right)\Gamma\left(\dfrac{\nu+3}{2}\right)}{\pi\nu^2\sigma^2{\Gamma\left(\dfrac{\nu}{2}\right)}^2}\dfrac{_2F_1\left(\dfrac{\nu+1}{2},\dfrac{\nu+3}{2};\dfrac{\nu}{2};\dfrac{\rho^i}{v_i(y^*_t)v_i(y^*_{t-i})}\right)}{v_i(y^*_t)^\frac{\nu+3}{2}v_i(y^*_{t-i})^\frac{\nu+1}{2}f(y_t,y_{t-i})},	
	\end{eqnarray*}	
	\begin{eqnarray*}
		\zeta_{ti\,2}=\dfrac{2(1-\rho^i)^{\frac{\nu}{2}+2}\Gamma\left(\dfrac{\nu+1}{2}\right)\Gamma\left(\dfrac{\nu+3}{2}\right)}{\pi\nu^2\sigma^2{\Gamma\left(\dfrac{\nu}{2}\right)}^2}\dfrac{_2F_1\left(\dfrac{\nu+1}{2},\dfrac{\nu+3}{2};\dfrac{\nu}{2};\dfrac{\rho^i}{v_i(y^*_t)v_i(y^*_{t-i})}\right)}{v_i(y^*_t)^\frac{\nu+1}{2}v_i(y^*_{t-i})^\frac{\nu+3}{2}f(y_t,y_{t-i})},	
	\end{eqnarray*}	
	and 
	\begin{eqnarray*}
		\tau_{ti}=\dfrac{2\rho^i(1-\rho^i)^{\frac{\nu}{2}+1}\Gamma\left(\dfrac{\nu+3}{2}\right)^2}{\pi\nu^2\sigma^2{\Gamma\left(\dfrac{\nu}{2}\right)^2}}\dfrac{_2F_1\left(\dfrac{\nu+3}{2},\dfrac{\nu+3}{2};\dfrac{\nu}{2}+1;\dfrac{\rho^i}{v_i(y^*_t)v_i(y^*_{t-i})}\right)}{\left(v_i(y^*_t)v_i(y^*_{t-i})\right)^\frac{\nu+3}{2}f(y_t,y_{t-i})},	
	\end{eqnarray*}	
	where $_2F_1(\cdot,\cdot;\cdot,\cdot)$ is the hypergeometric function previously defined, $f(y_t,y_{t-i})$ is given in Proposition \ref{p:jointdensity}, $y^*_t=y_t-\mu_t$, and $v_i(y^*_t)=1+\dfrac{1-\rho^i}{\nu\sigma^2}{y^*_t}^2$.
\end{proposition}	

The $Q$-function taking into account all the observations is then given by
\begin{eqnarray}\label{Qfunction}
Q(\boldsymbol\theta;\boldsymbol\theta^{(r)})=\sum_{t=m+1}^n\sum_{i=1}^mQ_{ti}(\boldsymbol\theta;\boldsymbol\theta^{(r)}).
\end{eqnarray}

With the $Q$-function completely specified, we are now able to perform the M-step of the CLEM algorithm.

\subsubsection*{CL M-step}

The CL-M step of the proposed algorithm consists of maximizing the $Q$-function given in (\ref{Qfunction}). It is worth noting that due to the decomposition of the $Q_{ti}$-function, the parameters $(\boldsymbol\beta,\sigma^2)$ are estimated independently from $\rho$ in each step of the CLEM-algorithm.

By taking $\partial Q(\boldsymbol\theta;\boldsymbol\theta^{(r)})/\partial (\boldsymbol\beta,\sigma^2)=0$, we obtain explicit solutions for the CLEM estimators for $(\boldsymbol\beta,\sigma^2)$ in the $(r+1)$th loop of the algorithm:
\begin{eqnarray}\label{clem_beta}
\sum_{t=m+1}^n\sum_{i=1}^m\zeta_{t-i\,i}^{(r)}(Y_{t-i}-{\bf x}^\top_{t-i}\boldsymbol\beta^{(r+1)})x_{t-i\,j}=-\sum_{t=m+1}^n\sum_{i=1}^m\zeta_{ti}^{(r)}(Y_t-{\bf x}^\top_{t}\boldsymbol\beta^{(r+1)})x_{tj},\quad j=1,\ldots,p,\nonumber
\end{eqnarray}
and
\begin{eqnarray}\label{clem_sigma2}
{\sigma^2}^{(r+1)}=\dfrac{1}{2m(n-m)}\bigg\{\sum_{t=m+1}^n\sum_{i=1}^m\zeta_{ti}^{(r)}(Y_t-{\bf x}^\top_{t}\boldsymbol\beta^{(r+1)})^2+\sum_{t=m+1}^n\sum_{i=1}^m\zeta_{t-i\,i}^{(r)}(Y_{t-i}-{\bf x}^\top_{t-i}\boldsymbol\beta^{(r+1)})^2
\bigg\}.
\end{eqnarray}

The maximization of the $Q$-function with respect to the parameter $\rho$ has associated score function given by 
%\textcolor{blue}{
%\begin{eqnarray*}
%&&\frac{\partial Q(\boldsymbol\theta;\boldsymbol\theta^{(r)})}{\partial\nu}=\sum_{t=m+1}^n\sum_{i=1}^m\bigg\{
%\log\left(\frac{\nu}{2(1-\rho^i)}\right)-\frac{3}{2}\Psi\left(\frac{\nu}{2}\right)\\
%&&+\left(\frac{2\tau_{ti}^{(r)}}{\nu}+1\right)-\dfrac{\zeta^{(r)}_{ti}+\zeta^{(r)}_{t-i\,i}}{2(1-\rho^i)}+
%\dfrac{\delta^{(r)}_{ti}+\delta^{(r)}_{t-i\,i}+\kappa^{(r)}_{ti}}{2}
%\bigg\}
%\end{eqnarray*}	
%and
\begin{eqnarray}\label{clem_rho}
\dfrac{\partial Q(\boldsymbol\theta;\boldsymbol\theta^{(r)})}{\partial\rho}=\sum_{t=m+1}^n\sum_{i=1}^m\dfrac{i\rho^{i-1}}{1-\rho^i}\bigg\{\nu+\tau^{(r)}_{ti}\left(2+\dfrac{1}{\rho^i}\right)-\dfrac{\nu(\zeta^{(r)}_{ti}+\zeta^{(r)}_{t-i\,i})}{2(1-\rho^i)}-\dfrac{\eta^{(r)}_{ti}}{1-\rho^i}
\bigg\}.
\end{eqnarray}	
%where $\Psi(x)=d\log\Gamma(x)/dx$ is the digamma function.

\begin{remark}
	A closed-form solution for $\boldsymbol\beta^{(r+1)}$ can be obtained by writing (\ref{clem_beta}) in a matrix form and, consequently, we also obtain ${\sigma^2}^{(r+1)}$ explicitly from (\ref{clem_sigma2}). Regarding the estimation of $\rho$, note that the optimization necessary to find the root of the gradient function (\ref{clem_rho}) is much simpler (it does not involve any complicated function) than the one in the direct maximization of (\ref{eq:cl_log_likelihood}). In particular, for the case $m=1$ (considered in both simulated and real data in this paper), the problem relies on finding the admissible root of a quadratic polynomial. This illustrates the drastic complexity reduction of using the proposed CLEM algorithm when facing numerical issues in the direct maximization of the composite likelihood function in (\ref{eq:cl_log_likelihood}).
\end{remark}

With the expectation and maximization steps described above, the CLEM algorithm is performed by iterating both steps until some suitable convergence criterion is achieved. The complete algorithm including the stopping criteria is given in Algorithm \ref{alg:clem}. 

\begin{remark}\label{nu_known}
	Note that in our CLEM algorithm we assume the $\nu$ parameter as known and provide steps to estimate the remaining parameters. In the finite sample situations explored in Sections \ref{sec:simulation} and \ref{sec:application}, we observed that the $\nu$ parameter is estimated very well by (a) direct optimization of the composite likelihood in \eqref{eq:cl_log_likelihood} and (b) initial guess obtained by maximizing \eqref{eq:student_t_initial_guess_lkl}. In our data analysis in Section \ref{sec:application}, we present estimation results of the $\nu$ parameter based on the above two techniques.   
\end{remark}

\begin{algorithm}
\footnotesize
	\caption{CLEM algorithm for the NSVt model}
	\label{alg:clem}
	\begin{enumerate}
		\item Start with some initial guess for the parameter vector, say $\boldsymbol\theta^{(0)}$.
		\item (CL E-step) Update the conditional expectations $\zeta^{(r)}_{ti\,1}$, $\zeta^{(r)}_{ti\,2}$, and $\tau_{ti}^{(r)}$ given in Proposition \ref{E-step}, with $\boldsymbol\theta^{(r)}$ being the CLEM estimate of $\boldsymbol\theta$ in the $r$th loop of the algorithm.
		\item (CL M-step) Find the CLEM estimates in the $(r+1)$th loop of the algorithm, say $\boldsymbol\theta^{(r+1)}$, by maximizing the $Q$-function defined in (\ref{Qfunction}).
		\item Check if some pre-specified stopping criterion is satisfied, for example if $\|\boldsymbol\theta^{(r+1)}-\boldsymbol\theta^{(r)}\|<\epsilon$ or $\|Q(\boldsymbol\theta^{(r+1)};\boldsymbol\theta^{(r)})-Q(\boldsymbol\theta^{(r)};\boldsymbol\theta^{(r)})\|<\epsilon$, for some small $\epsilon>0$. If the stopping criterion is satisfied, then the CLEM estimate of $\boldsymbol\theta$ is $\boldsymbol\theta^{(r+1)}$, which finishes the algorithm. Otherwise, update $\boldsymbol\theta^{(r)}$ by $\boldsymbol\theta^{(r+1)}$ and come back to the Step 2.
	\end{enumerate}
\end{algorithm}

\begin{remark}
	For a general state-space model without the inclusion of covariates, \cite{var08_2} review consistency and asymptotic normality of the estimators obtained using a pairwise likelihood approach. In their work,  certain regularity conditions, such as unbiasedness along with a stationarity assumption of the latent state process, are placed on the score function to establish consistency and asymptotic normality.  The work in \cite{ngetal2011} also discusses large sample properties of estimators from a composite likelihood method. With the latent factor being Gaussian autoregressive, under standard regularity conditions on the parameter space and also with certain moment assumptions, they establish consistency of the estimators when the covariates are stationary, bounded, and $m$-dependent processes. Following the previous works and with similar assumptions placed, large sample properties of the estimators from our CLEM method can be established. Simulated results to be presented in Section \ref{sec:simulation} show empirical evidence of consistency and asymptotic normality of the CLEM estimators as the sample size grows.  
\end{remark}

As argued by \cite{gaoson2011}, it is cumbersome to assess the standard errors of the CLEM estimates via the Godambe information. The authors suggest the usage of resampling methods to achieve this aim. To address that, we propose a parametric bootstrap as follows. Generate $B$ bootstrap replications from our fitted (via CLEM algorithm) NSVt process, say $\{\widetilde Y_{1j},\ldots$ $,\widetilde Y_{nj};\,j=1,\ldots,B\}$. Then, compute the CLEM estimates for each generated trajectory, say $\widetilde{\boldsymbol{\theta}}_j$, for $j=1,\ldots,B$. Finally, we use either the empirical distribution or normal approximation of the bootstrap estimates to assess the standard errors and to construct confidence intervals. In our real data analysis in Section \ref{sec:application}, we provide results on finding quantiles and standard errors based on this bootstrap procedure. We discuss there the usefulness of this procedure in determining the statistical significance of the regression coefficients.

\section{Monte Carlo simulation studies}\label{sec:simulation}

In this section, we investigate the finite-sample performance of the composite likelihood estimators using the EM-algorithm approach (expressed as CLEM) discussed in Section \ref{sec:clem_section}. Also, we present the results obtained by the direct optimization of the composite log-likelihood in \eqref{eq:cl_log_likelihood} (expressed as CL).  

We perform Monte Carlo simulations of the NSVt model at various parameter choices and sample sizes. We simulate from the NSVt model in Definition \ref{def:tar} under the following setups: (i) sample sizes $n \in \{ 100,300,500,1000 \}$; (ii) $\nu \in \{ 3,5,8 \}$ and $\rho \in \{ 0.5, 0.7, 0.9 \}$; and (iii) $p \in \{ 5, 10 \}$ and $\sigma^2 \in \{ 0.5 , 1 \}$. 

The $p$-dimensional covariate series ${\bf x}_t$ are generated from independent zero-mean ARMA(1,1) processes (white noises following standard normal distributions), where the ARMA coefficients were chosen at random from the Uniform distribution $U(0.3,0.7)$. After generating the covariate vector, it is kept fixed during the 500 simulation runs. The regression coefficients $\beta_1,\ldots,\beta_p$ are also chosen at random from a $U(-1,1)$ distribution and then fixed for the simulation runs. The above simulation setup is repeated 500 times and the estimates of ${\boldsymbol\beta}$, $\sigma^2$, and $\rho$ are presented from these runs. 

To assess the performance in estimating ${\bf \beta}$, we consider the relative $L_{2}$ error given by $D(\widehat{\boldsymbol\beta} ,{\boldsymbol\beta}) = \dfrac{ \| \widehat{\boldsymbol\beta} - {\boldsymbol\beta} \|^2}{ \| {\boldsymbol\beta} \|^2 }$, where $\widehat{\boldsymbol\beta}$  and ${\boldsymbol\beta}$ denote the estimated and true regression coefficient vectors respectively, and $ \|\cdot \|$ is the $L_2$ norm. In evaluating the performance of estimators of the parameters $\sigma^2$ and $\rho$, we present boxplots of the estimates along with the corresponding true values. We provide results from the composite likelihood EM-algorithm approach discussed in Section \ref{sec:clem_section} as well as those obtained by the direct optimization of the composite log-likelihood function in \eqref{eq:cl_log_likelihood}.

In the left panels in Figures \ref{fig:f1}-\ref{fig:f4}, we present the boxplots of the relative errors on the estimation of the regression coefficients $\beta_j$, $j=1,2,3,4,5$ ($p=5$ case), at the indicated sample sizes using the measure $D(\widehat{\boldsymbol\beta} ,{\boldsymbol\beta})$. These boxplots are based on 500 runs of the NSVt model. We observe that for both CLEM and CL methods, the errors decrease as the sample size grows and we also notice a comparable performance between the two methods. In the right panels in Figures \ref{fig:f1}-\ref{fig:f4}, we provide the boxplots of the estimates of $\sigma^2$ and $\rho$ along with the true values plotted using dashed lines. Here again for both CLEM and CL methods, as the sample size grows, the boxplots get narrower and the median gets closer to the true value. Regarding the estimation of the parameter $\rho$, which controls the persistence of volatility, we observe  a better performance of the CLEM approach over the CL method. It must be mentioned that in estimating $\sigma^2$ and $\rho$, we witnessed a similar performance when the number of regression coefficients $p$ is increased to 10 and also when $\sigma^2$ was set to 1. We have provided similarly generated boxplots for other parameter choices in the Supplementary Material.

\begin{figure*}
	\begin{subfigure}{.5\textwidth}
		\includegraphics[scale=0.38]{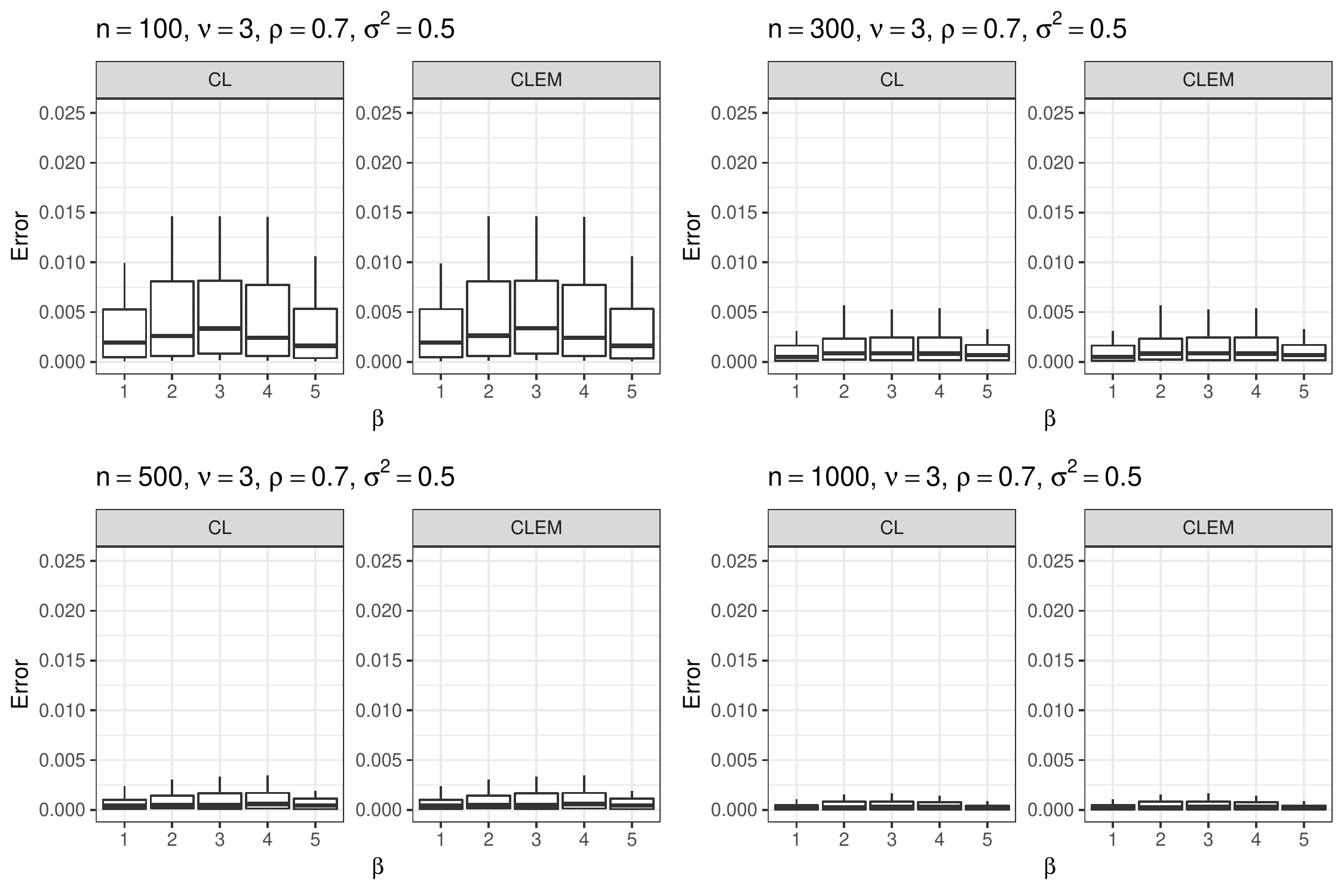}
	\end{subfigure}
	\begin{subfigure}{.5\textwidth}
		\includegraphics[scale=0.38]{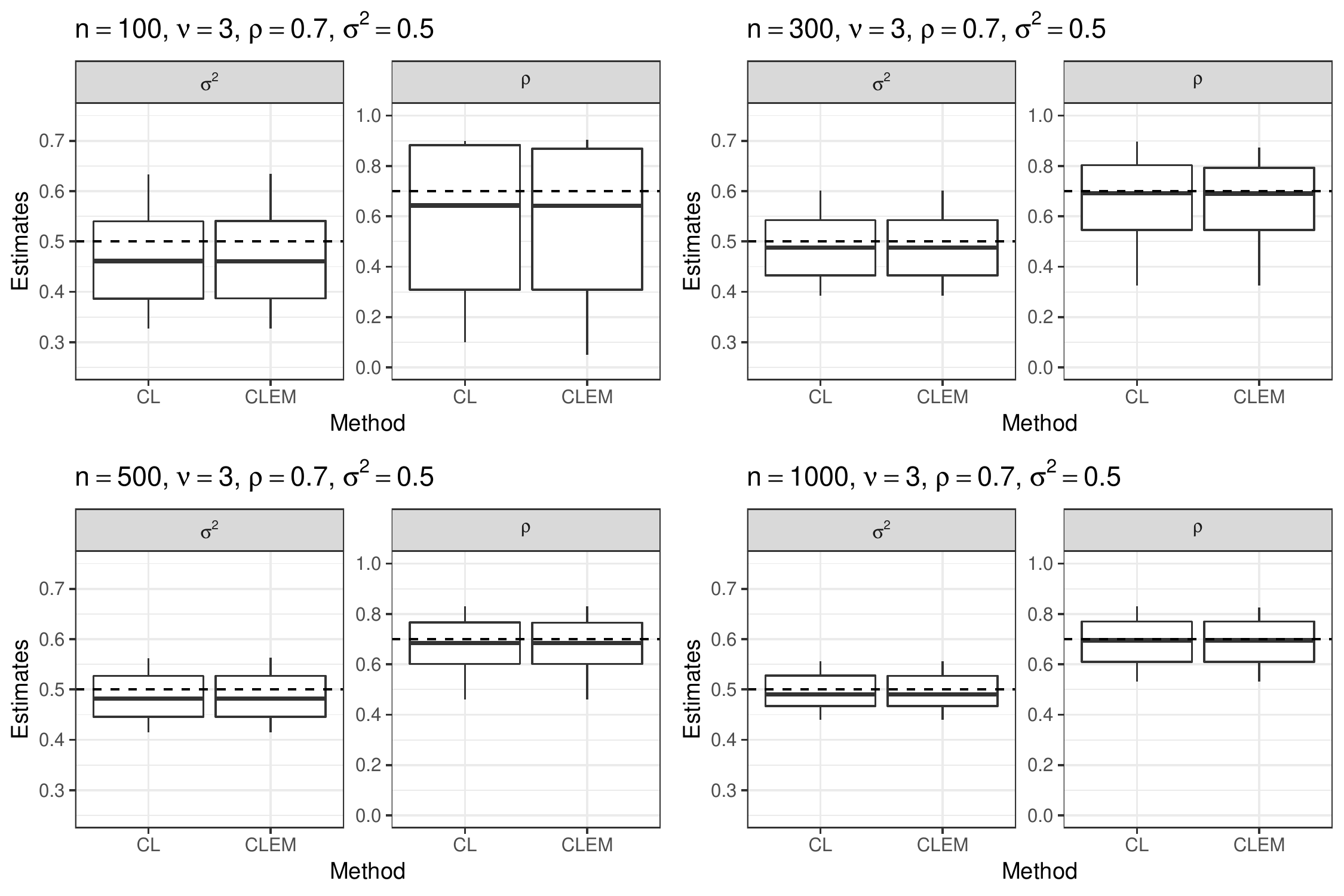}
	\end{subfigure}
	\caption{Left panel: Error boxplots for estimating the regression coefficients $\beta_j$, $j=1,2,3,4,5$, at the indicated sample sizes using the measure $D(\widehat{\boldsymbol\beta} ,{\boldsymbol\beta})$. Right panel: Boxplots of estimates of $\sigma^2$ and $\rho$ at the indicated sample sizes. True values are given via dashed lines. }  \label{fig:f1}
\end{figure*}

\begin{figure*}
	\begin{subfigure}{.5\textwidth}
		\includegraphics[scale=0.38]{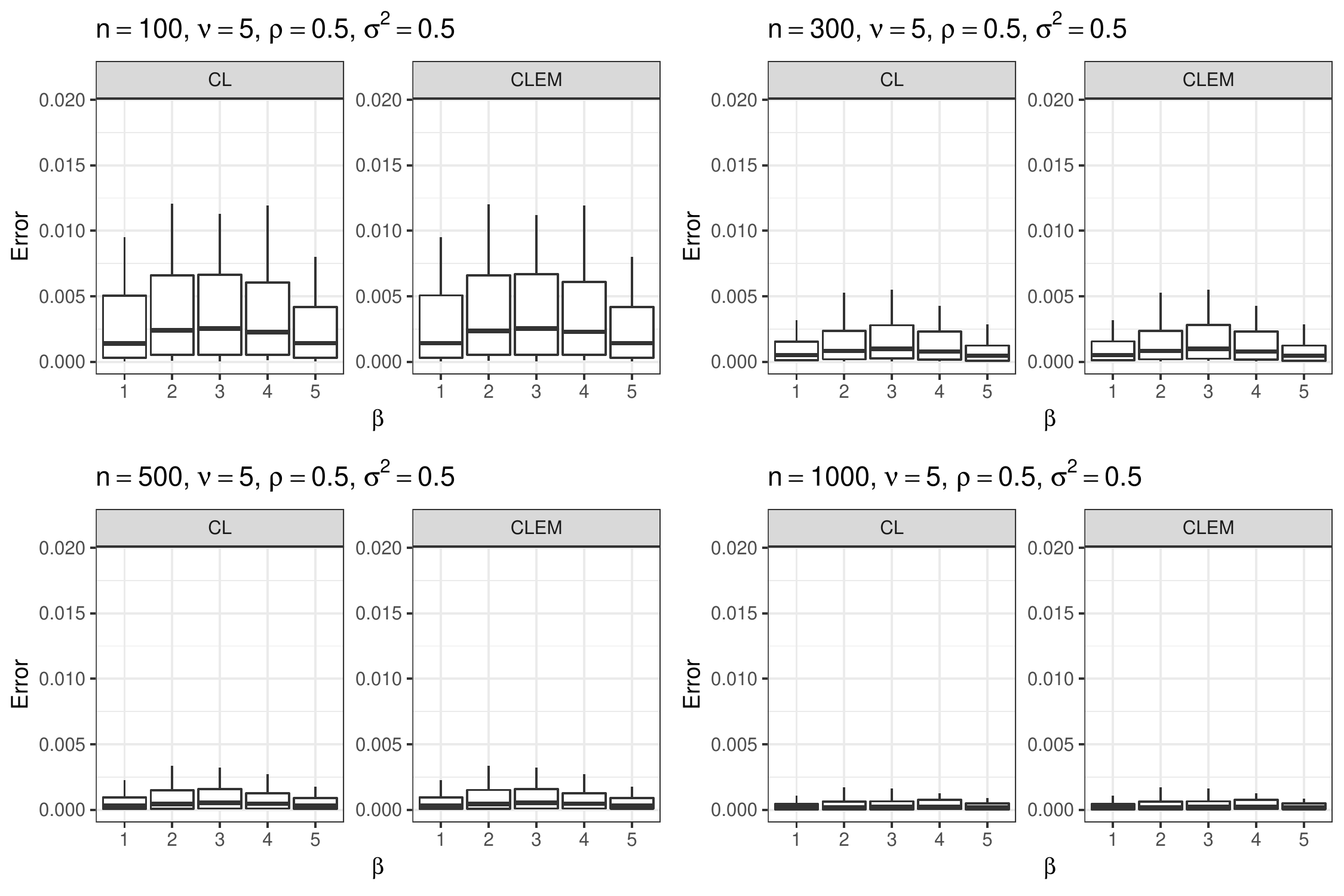}
	\end{subfigure}
	\begin{subfigure}{.5\textwidth}
		\includegraphics[scale=0.38]{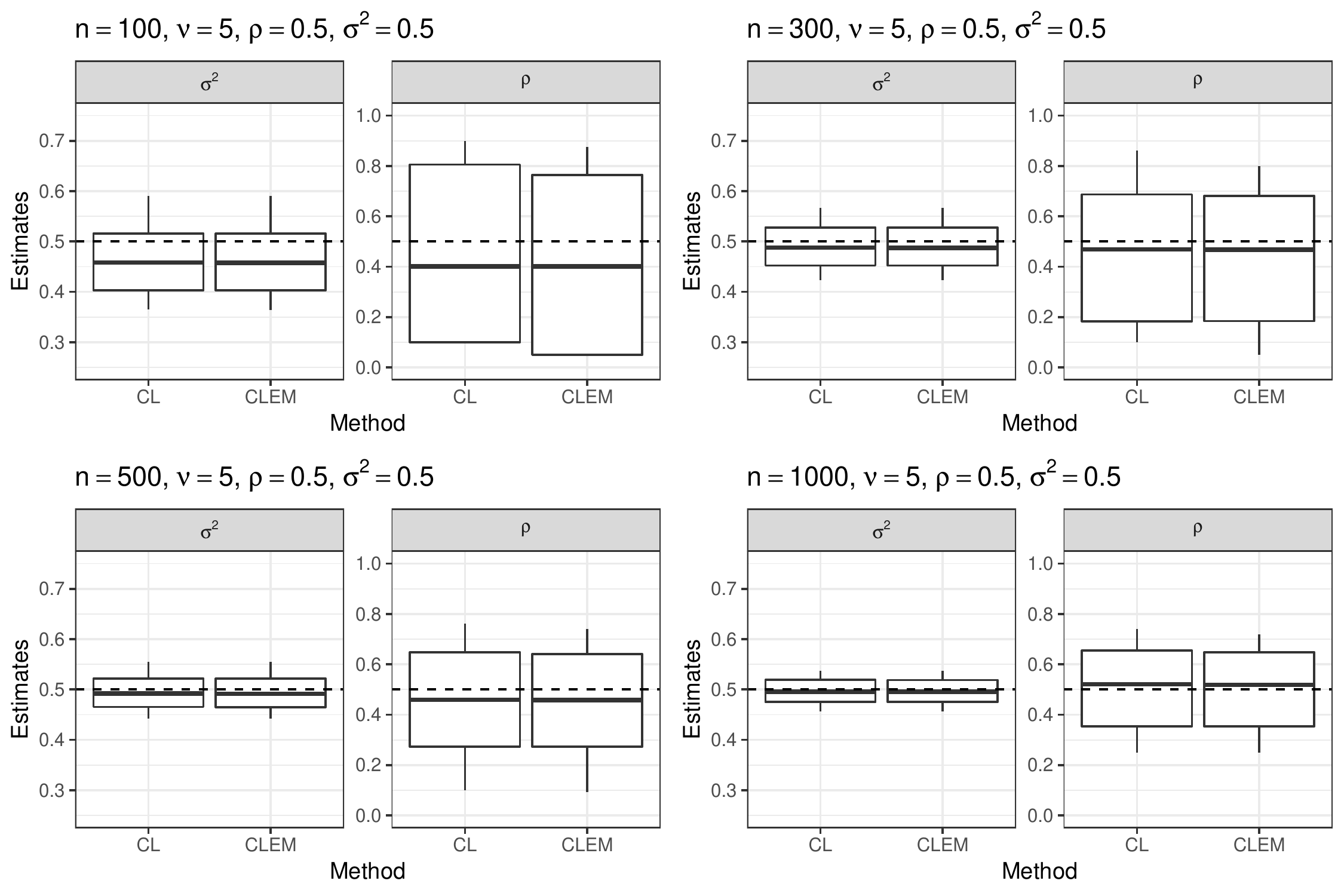}
	\end{subfigure}
	\caption{ Left panel: Error boxplots for estimating the regression coefficients $\beta_j$, $j=1,2,3,4,5$,  at the indicated sample sizes using the measure $D(\widehat{\boldsymbol\beta} ,{\boldsymbol\beta})$. Right panel: Boxplots of estimates of $\sigma^2$ and $\rho$ at the indicated sample sizes. True values are given via dashed lines.  } \label{fig:f2}
\end{figure*}

\begin{figure*}
\begin{subfigure}{.5\textwidth}
\includegraphics[scale=0.38]{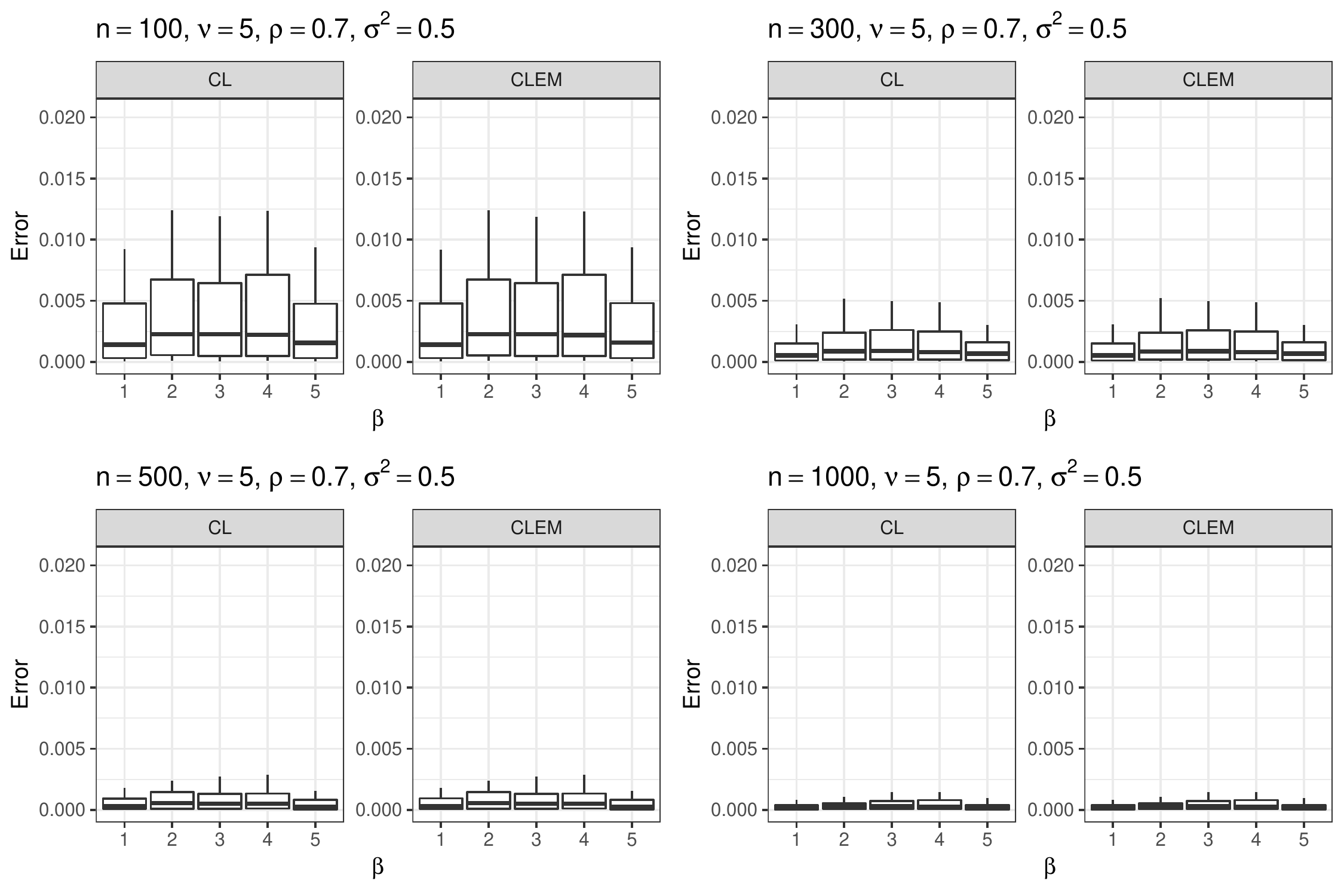}
\end{subfigure}
\begin{subfigure}{.5\textwidth}
\includegraphics[scale=0.38]{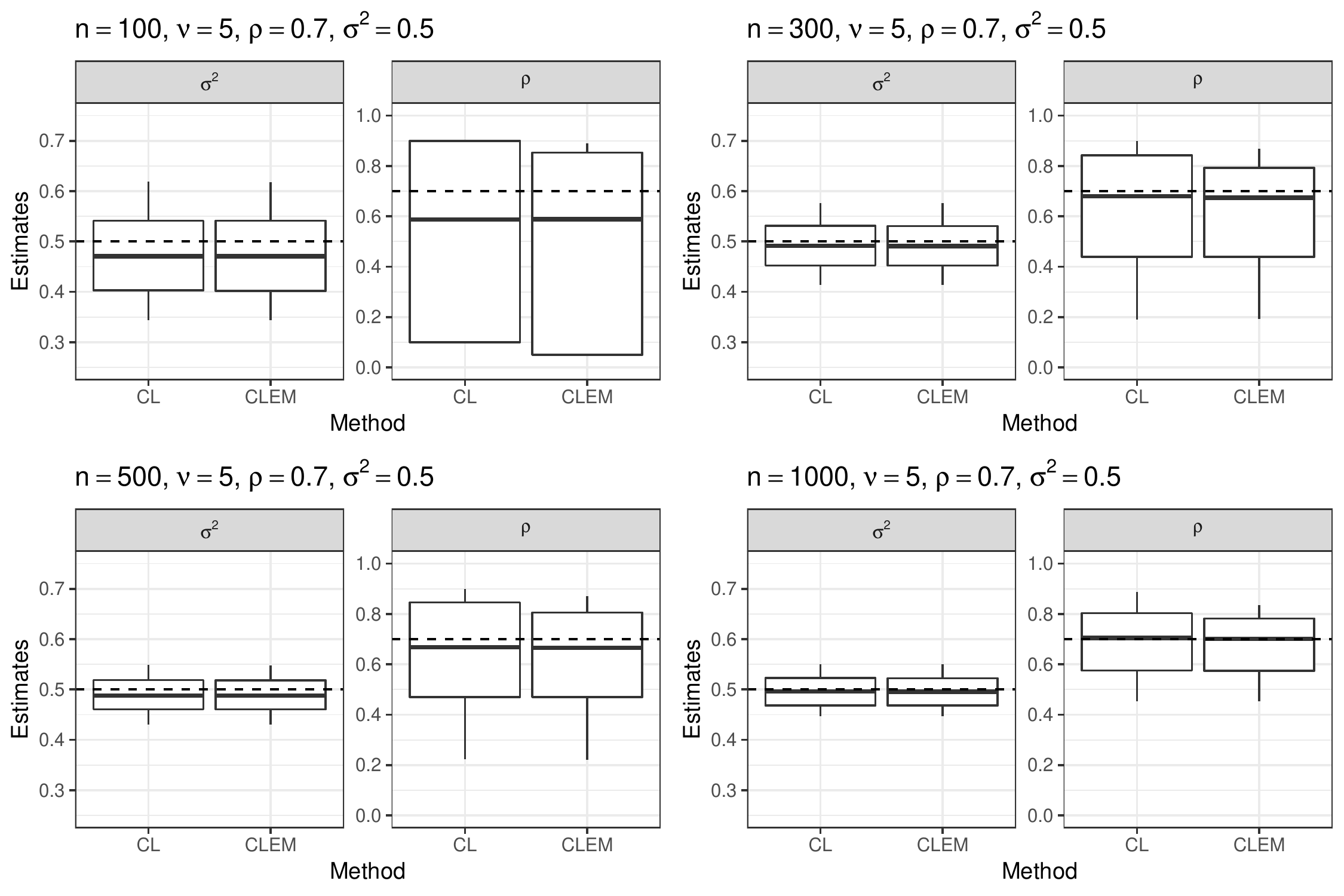}
\end{subfigure}
\caption{ Left panel: Error boxplots for estimating the regression coefficients $\beta_j$, $j=1,2,3,4,5$,  at the indicated sample sizes using the measure $D(\widehat{\boldsymbol\beta} ,{\boldsymbol\beta})$. Right panel: Boxplots of estimates of $\sigma^2$ and $\rho$ at the indicated sample sizes. True values are given via dashed lines.  } \label{fig:f3}
\end{figure*}

\begin{figure*}
\begin{subfigure}{.5\textwidth}
\includegraphics[scale=0.38]{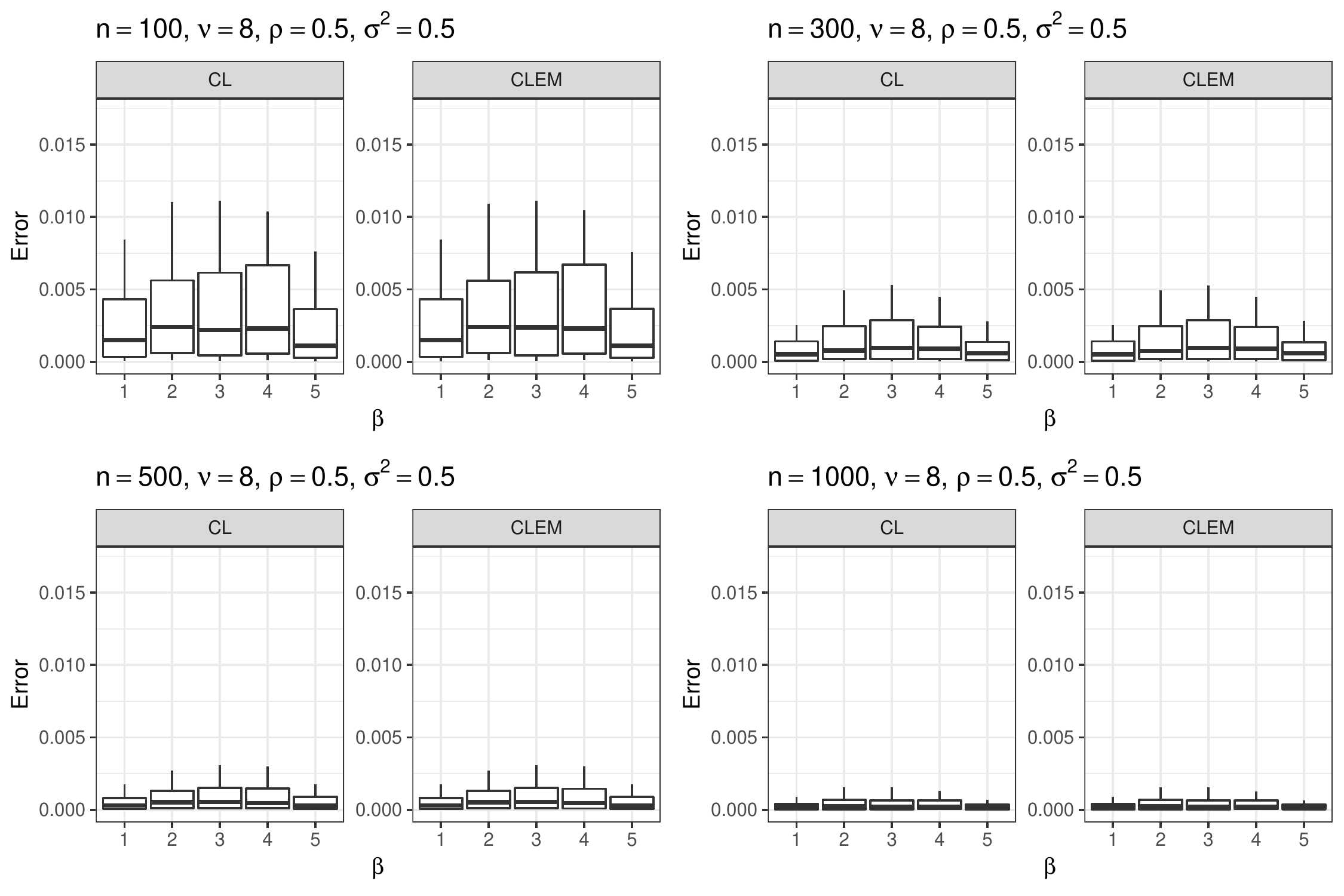}
\end{subfigure}
\begin{subfigure}{.5\textwidth}
\includegraphics[scale=0.38]{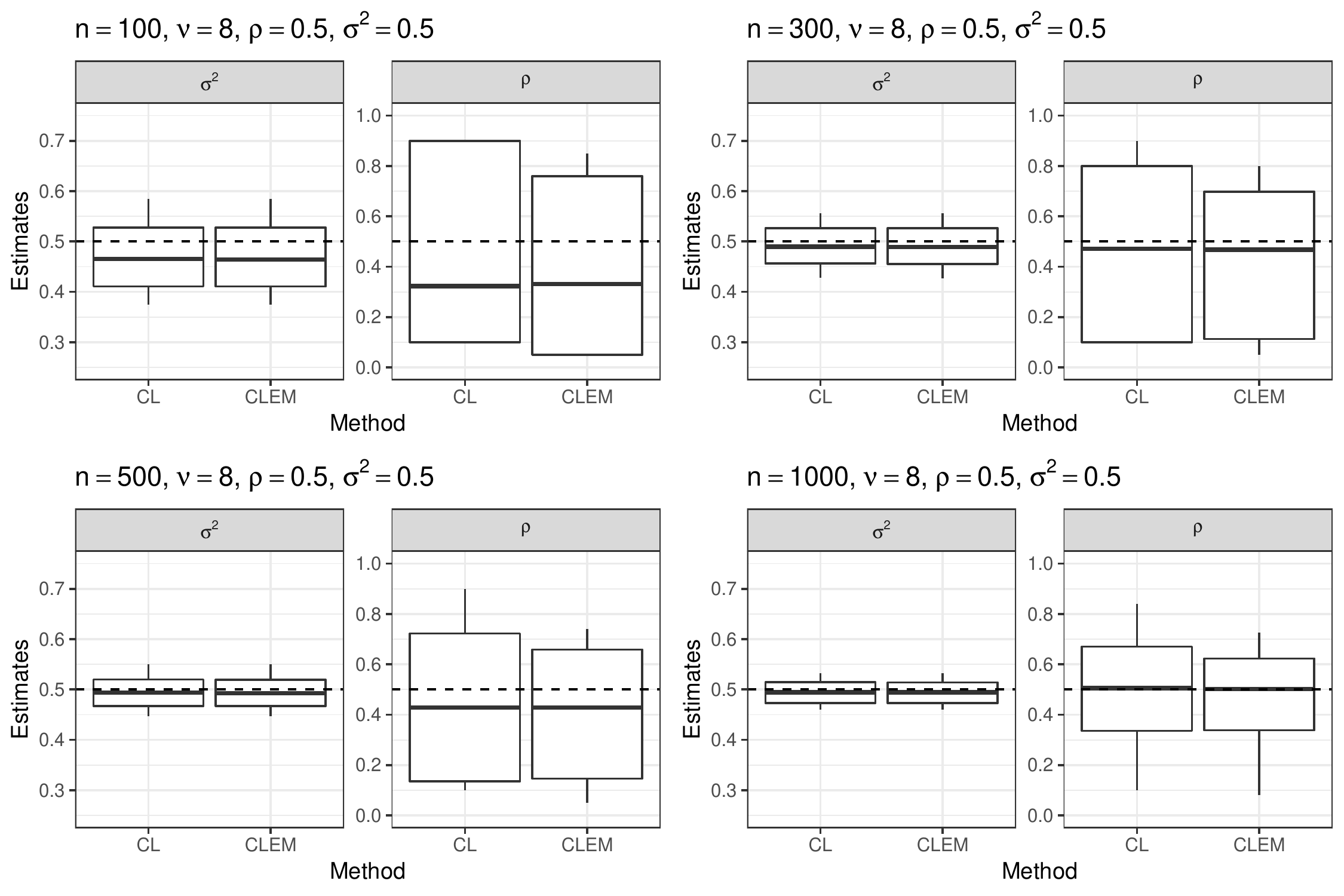}
\end{subfigure}
\caption{ Left panel: Error boxplots for estimating the regression coefficients $\beta_j$, $j=1,2,3,4,5$,  at the indicated sample sizes using the measure $D(\widehat{\boldsymbol\beta} ,{\boldsymbol\beta})$. Right panel: Boxplots of estimates of $\sigma^2$ and $\rho$ at the indicated sample sizes. True values are given via dashed lines. } \label{fig:f4}
\end{figure*}

\section{Application to exchange-traded funds based portfolio management} \label{sec:application}

We look at an application of our method in finding optimal portfolios using the US sector ETFs (Exchange-Traded Funds) data. An ETF constitutes a combination or basket of securities, such as stocks, that is traded on an exchange. One example is the SPDR S\&P 500 ETF\footnote{\url{https://finance.yahoo.com/quote/SPY/}} that tracks the S\&P 500 index. A sector ETF is a combination of stocks restricted to a particular sector of the economy; for instance, the Health Care Select Sector SPDR Fund\footnote{\url{https://finance.yahoo.com/quote/xlv?ltr=1}} tracks stocks in the health care sector. \cite{avellaneda10} discuss statistical arbitrage strategies using sector ETFs as risk factors. In the previous work, the regression of stock price returns $R_{i,t} $, $i=1,2,\hdots,M$, of $M$ companies on $p$ active US sector ETF returns $S_{i,t}$, $i=1,2,\hdots, p$, is carried out. Sector ETFs in this regression setup are viewed as proxies of risk factors for the individual companies. Such a regression assists in managing portfolios and more importantly, getting to a market-neutral portfolio wherein the returns on the portfolio are more protected from shocks in the sectors of the economy. More precisely, the regression setup is
\vspace{-0.3cm}
\begin{equation}\label{eq:trading_reg}
R_{i,t} = \sum_{j=1}^p \beta_{i,j} S_{j,t} + \epsilon_{i,t}
\vspace{-0.3cm}
\end{equation}   
for $t=1,2\hdots,n$, $i=1,2,\hdots,M$, where the $\{\epsilon_{i,t}\}$'s are white noises. Knowledge of the $\beta_{i,j}$ helps reaching a market-neutral portfolio that satisfies $\sum_{i=1}^M \beta_{i,j} A_i=0$, $\forall j=1,2,\hdots,p$, where $A_i$ is the amount invested in $i^{th}$ company. Further, the $\beta_{i,j}$ coefficients enables a active portfolio management strategy wherein portfolio related decisions can be made more frequently. To obtain the coefficients in \eqref{eq:trading_reg},  \cite{avellaneda10} carry out a multiple regression analysis using sparsity/ridge type constraints on the coefficients. Similar to the setup in \eqref{eq:trading_reg}, \cite{lalloux2000} and \cite{pleroux2002} consider the regression of stock price returns on approximately $15$ factors to study movements in returns explained by a few important factors in the market. \cite{krauss17} and \cite{stubinger18} are related works that discuss pairs trading strategies using high-frequency returns data, where the exposure of the returns (response) from the chosen pairs against common risk factors (regressors) from the market is studied. The statistical significance of these risk factors in this regression is analyzed through the typical linear regression inference approach.

With our NSVt time series model given in Definition \ref{def:tar}, we take $Y_{i,t} = \overline{R}_{i,t}$, where $\overline{R}_{i,t}$ is the standardized stock price return of company $i$ at time $t$ and ${\bf x}_t = (x_{1,t},x_{2,t},\hdots,x_{p,t})^{\top}$, where $x_{j,t}=S_{j,t}$ is the standardized returns of the $j^{th}$ sector ETF. Standardization is done by subtracting the mean and dividing by the standard deviation over local time windows (see the description for training periods below). The $p=11$ sector ETFs considered are given in Table \ref{tab:etf_boot_coefficients}. We consider optimizing market-neutral portfolios involving $M=2$ stocks:  Chevron(CVX) and Exxon (XOM). Starting from June 3, 2019, we take the next 15 trading days as starting days. Each training period consists of 4 days from, and including, the corresponding starting day. We thus have a total of 15 training periods, each having $n=720$ observations. It must be noted that we only consider data during the hours 9:30 am to 3:30 pm in every trading day and therefore we have 2-minute returns (frequency of the data). The majority of the SV models literature involving applications of financial returns use lower frequency data, such as daily or monthly. We use high-frequency intra-day returns data and this assists in analyzing relationships over smaller time periods. This also enables making portfolio decisions more frequently.

In each training period, we fit the following time series regression models: 
\vspace{-0.3cm}
\begin{equation} \label{eq:etf_two_regressions}
(i) \;\; Y_{1,t} \sim (x_{1,t},x_{2,t},\hdots,x_{11,t})^\top, \;\; \;
(ii) \;\; Y_{2,t} \sim (x_{1,t},x_{2,t},\hdots,x_{11,t})^\top,
\vspace{-0.5cm}
\end{equation}
where $t=1,2,\hdots,n$ and $Y_{1,t}$ and $Y_{2,t}$ are the standardized returns of Chevron (CVX) and Exxon (XOM) respectively. We are assuming that $\{Y_{1,t}\}$ and $\{Y_{2,t}\}$ are conditionally independent NSVt processes given the covariate vector $(x_{1,t},x_{2,t},\hdots,x_{11,t})^\top$. The above regressions are fitted using high-frequency returns data\footnote{Data source: \url{https://polygon.io/}}. As an example, Figure \ref{fig:response_plot_example} shows plots of the standardized returns of the two response processes, $Y_{1,t}$ (Chevron CVX) and $Y_{2,t}$ (Exxon XOM), for the training period June 14-19 (4 trading days), 2019. 

\begin{figure*}
	\begin{center}
		\includegraphics[scale=0.44]{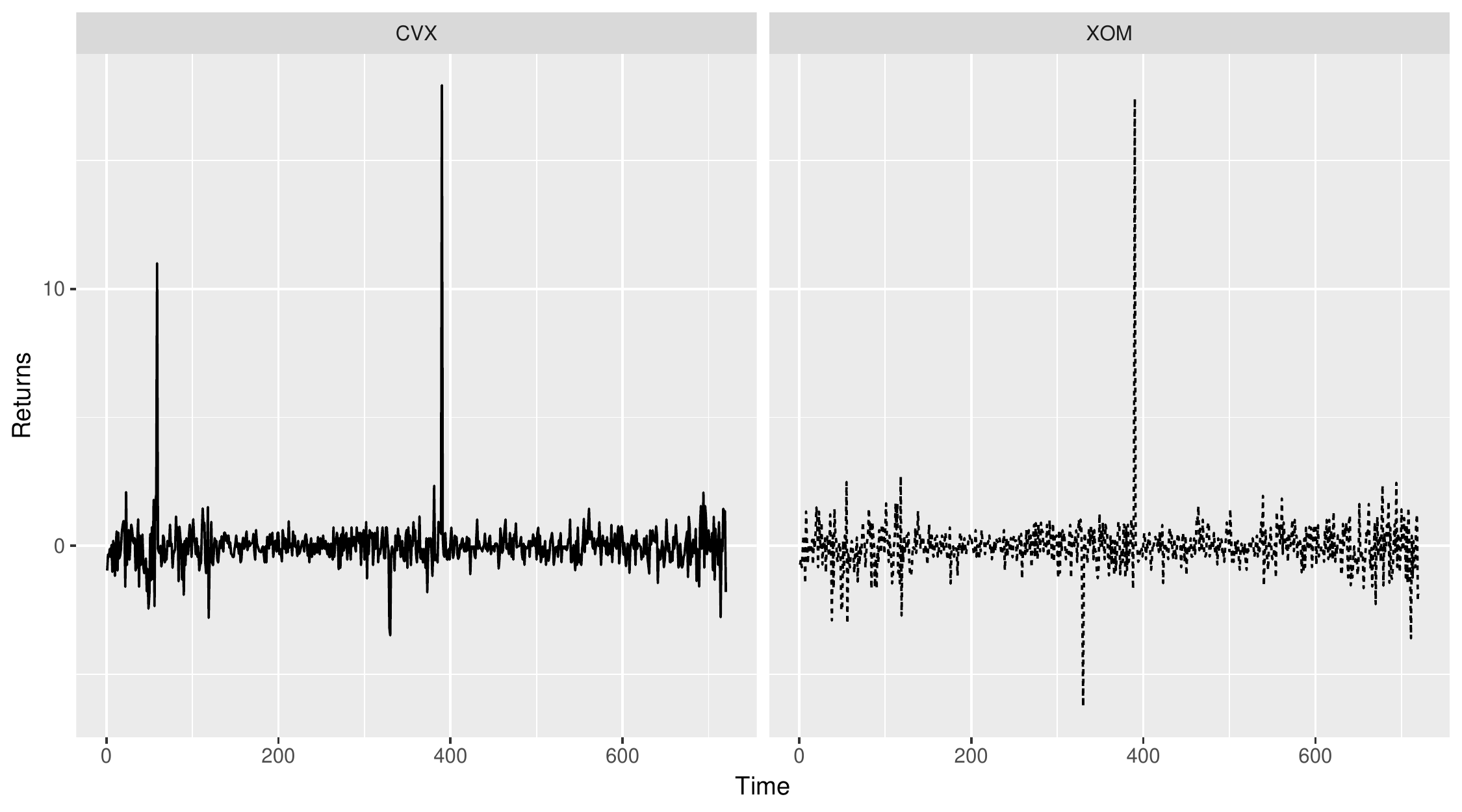}
		\vspace{-0.3cm}
		\caption{Plots of the standardized returns of Chevron CVX ($Y_{1,t}$) and Exxon XOM ($Y_{2,t}$) for the training period June 14-19, 2019.} \label{fig:response_plot_example}
	\end{center}
\end{figure*}

By fitting our NSVt model to the CVX and XOM time series, we obtain the estimates of regression coefficients and parameters $(\sigma^2, \rho)$ using the CLEM algorithm from Section \ref{sec:clem_section}. The estimates for the $\nu$ parameter were obtained via a). Direct optimization of the composite likelihood in \eqref{eq:cl_log_likelihood} and b). Initial guess obtained by maximizing the Student-$t$ log-likelihood function in \eqref{eq:student_t_initial_guess_lkl}. For comparison purposes, we also consider the approaches by \cite{avellaneda10}, where the regression coefficients are obtained from ridge and LASSO models (\citealp{hastie2016}) by assuming \eqref{eq:etf_two_regressions}. We aim to contrast the portfolio returns of our method with these competitors.

First, in Table \ref{tab:etf_boot_coefficients},  we list the 11 US sector ETFs that are considered for the regressions in \eqref{eq:etf_two_regressions} along with the average regression coefficients over the 15 training periods in June 2019. More precisely, the second column reports the average regression coefficients $\overline{\beta}_{1,j}$ estimated by our CLEM method when $Y_{1,t}$ (Chevron CVX) is the response. Each training period yields a set of $p=11$ regression coefficients and we present the 11 averages obtained over the 15 training periods. Similarly, the sixth column has the corresponding coefficients when $Y_{2,t}$ (Exxon XOM) is the response. We notice the importance of the coefficient corresponding to sector ETF XLE (second line in Table \ref{tab:etf_boot_coefficients}), which is the energy sector of the market\footnote{\url{https://finance.yahoo.com/quote/XLE/}}. We also notice that the coefficients behave similarly across CVX and XOM and this is expected as the two major stocks from the energy sector are expected to be impacted in the same way by movements in the sectors. The advantage of this regression of individual stock price returns on sector ETFs is that it tells us which sector's movements impact any given company's returns. 

\begin{table} 
\footnotesize
	\begin{center}
	\begin{tabular}{| c |c c c c|c c c c | }
		\hline
		& \multicolumn{4}{c}{CVX} & \multicolumn{4}{c|}{XOM} \\
		\cline{2-9}
		ETF & $\overline{\beta}_{1,j}$  & Lower & Upper & SE & $\overline{\beta}_{2,j}$ & Lower & Upper & SE \\
		\hline
		XLF Financial  & 0.024  & $-$0.013 & 0.060 & 0.019  & 0.014  & $-$0.035 &  0.065 & 0.026  \\
		XLE Energy & 0.765 & 0.719 & 0.805 & 0.022 & 0.780 & 0.692 & 0.797 &  0.027  \\
		OIH Oil Services & $-$0.076 & $-$0.106 & $-$0.034 & 0.018 & $-$0.086 & $-$0.121 & $-$0.035 & 0.022  \\ 
		XLK Tech & 0.040 & $-$0.002 & 0.090 & 0.023 & 0.059 & $-$0.004 &  0.119  &  0.031  \\
		XLP Consumer Staples & 0.059 & 0.015  & 0.106 & 0.023 & 0.069 & 0.009  & 0.121 & 0.029  \\
		XLV Health Care & 0.026 & $-$0.013 &  0.068 & 0.021 & 0.038 & $-$0.019 &  0.089 &  0.027    \\
		XLU Utilities & 0.052 & 0.011 & 0.092 & 0.021 & 0.021 & $-$0.028 & 0.083 &  0.028  \\
		GDX Gold Miners & $-$0.023 & $-$0.053 & 0.013 & 0.017 & $-$0.020 & $-$0.056 & 0.026  &  0.021   \\
		XLI Industrials & $-$0.049 & $-$0.097 & $-$0.003 & 0.024 & $-$0.056 & $-$0.114 &  0.013 & 0.032   \\
		IYE Energy Ishare & 0.006 & $-$0.014 & 0.024 & 0.010 & 0.005 & $-$0.021 &  0.032  &  0.013  \\
		XME Metals & 0.000 & $-$0.019 & 0.019 & 0.010 & 0.014 & $-$0.010  & 0.043  & 0.013  \\
		\hline
	\end{tabular} 
	\end{center}
	\vspace{-0.3cm}
	\caption{$p=11$ US Sector ETFs, the average regression coefficient estimates $\overline{\beta}_{1,j}$ and $\overline{\beta}_{2,j}$  obtained over the 15 training periods in June 2019, the 95\% Lower and Upper confidence limits of the regression coefficient estimates. Limits and standard errors (SE) obtained using parametric bootstrap. Results provided for the two companies Chevron (CVX) and Exxon (XOM) from the periods in June 2019.} \label{tab:etf_boot_coefficients}
\end{table}

In Table \ref{tab:etf_boot_coefficients}, we also present the 95\% confidence interval limits of the estimated regression coefficients of the two time series regression models in \eqref{eq:etf_two_regressions}. Here we employ a parametric bootstrap procedure, where replications from the fitted NSVt model are generated. Histograms of the bootstrap estimates of these $p=11$ regression coefficients, provided in the Supplementary Material, indicate a Gaussian-like behavior and we then compute the confidence interval limits based on the normal-bootstrap approach. For the two companies Chevron (CVX) and Exxon (XOM), the confidence limits and standard errors provided in this table are averages of the confidence limits and standard errors obtained over the 15 training periods from June 2019. We again notice the importance of the energy sector ETF XLE. The sector ETF OIH (Oil Services) is seen to have a negative relationship with the returns of the two companies. We also observe from the confidence limits that most other sector ETFs are not significant in explaining movements in the returns of the two companies. For the CL method which directly optimizes the composite likelihood in \eqref{eq:cl_log_likelihood}, an analogous table is provided in the Supplementary Material.

Next, in the left plot in Figure \ref{fig:etf_nu_returns}, we present estimates of the $\nu$ parameter from Definition \ref{def:tar} over the 15 training periods in June 2019.  We notice a `co-movement' type behavior in this $\nu$ parameter for these two stocks across this period. Changes in the $\nu$ parameter across the days in June 2019 can also indicate periods of high and low volatility. As an example, for Chevron (CVX), higher values of $\nu$ seen in the early part of that month can indicate a time of low volatility whereas a decreasing $\nu$ in the latter part of that month indicates a period of higher volatility. To add more evidence to this observation, in the Supplementary Material, we provide plots of the estimated conditional variance $\mbox{Var}(Y_{t+1}|Y_t=y_t)$ obtained from Proposition \ref{p:condvar}. There, we observed that the conditional variances were notably higher for the training periods in the latter part of June 2019. Note that the $\nu$ parameter is obtained by maximizing the Student-$t$ log-likelihood function in \eqref{eq:student_t_initial_guess_lkl}. A similar plot of estimates of $\nu$ obtained by the direct optimization of the composite likelihood in \eqref{eq:cl_log_likelihood} is provided in the Supplementary Material.  Plots of the estimates of $\sigma^2$ and $\rho$ using the CLEM-algorithm during the 15 training periods in June 2019 using the CLEM and CL methods are also given in the Supplementary Material.

In Table \ref{tab:etf_boot_oparam}, we provide the confidence limits and standard errors of the remaining parameters and their estimates respectively under our stochastic volatility model. The confidence limits of the $\nu$ parameter along with its plot in Figure \ref{fig:etf_nu_returns} indicate non-Gaussian characteristics in the response process which is in agreement with several works in the literature witnessing a heavy tail behavior in stock price returns data. For the CL method, an analogous table is provided in the Supplementary Material. To assess the model adequacy, we use the probability integral transform where a histogram of the distribution function transform of the estimated values of the response is presented. These histograms, obtained for the various training periods under the NSVt model, are provided in the Supplementary Material. These plots show that the proposed Student-$t$ process is well-fitted to the time series considered in this application.

\begin{table} 
\footnotesize
	\begin{center}
	\begin{tabular}{| c | c c c c|c c c c| }
		\hline
		& \multicolumn{4}{c}{CVX} & \multicolumn{4}{c|}{XOM} \\
		\cline{2-9}
		& Lower & Upper & SE & Mean & Lower & Upper & SE & Mean  \\
		\cline{2-9}
		$\sigma^2$  & 0.039 &  0.064 & 0.007 & 0.049 &  0.068 & 0.116 & 0.012 & 0.086  \\
		$\nu$  & 4.036 &  4.180 & 0.061 & 4.089 & 3.428 &  3.472 & 0.010 & 3.432  \\
		$\rho$  & 0.490 & 0.869 & 0.098 & 0.757 &  0.562 & 0.892  & 0.089 & 0.793  \\ 
		\hline
	\end{tabular}  
	\end{center}
	\vspace{-0.3cm}
	\caption{Average lower and upper 95\% confidence interval limits for the parameters $\sigma^2$, $\nu$, and $\rho$. Limits, standard errors (SE) and Mean obtained using a parametric bootstrap approach. Results provided for the two companies Chevron (CVX) and Exxon (XOM) from the periods in June 2019.} \label{tab:etf_boot_oparam} 
\end{table}

With the regression coefficients $\beta_{i,j}$ for $i=1,2$ and $j=1,2,\hdots,p=11$ from \eqref{eq:etf_two_regressions}, we obtain portfolio weights using the optimization problem in \eqref{eq:portfolio_optimization}. Here we take the daily returns $R_{1,t}$ (CVX) and $R_{2,t}$ (XOM) from the 4 days in training period considered for \eqref{eq:etf_two_regressions} plus an additional 2 trading days. After each training period, the following 2 trading days is treated as our evaluation period. We thus have 15 evaluation periods to assess the performance of our method. We consider the optimization procedure given by
\vspace{-0.3cm}
\begin{equation} \label{eq:portfolio_optimization}
\underset{w_{1,j},w_{2,j}}{\text{maximize}}
  \sum_{t \in P_{tr.}} w_{1,j} R_{1,t} + w_{2,j} R_{2,t}, \; \;
\text{subject to} \;\;
 w_{1,j} + w_{2,j} = 1, \;\; \mbox{and} \; \;
 w_{1,j}\beta_{1,j} + w_{2,j}\beta_{2,j} < \delta,
 \vspace{-0.3cm}
\end{equation}
where $w_{1,j}$ and $w_{2,j}$ are the respective investment weights corresponding to the two stocks, CVX and XOM respectively, for the $j^{th}$ portfolio. $R_{1,t}$ and $R_{2,t}$ are the daily returns of the two stocks CVX and XOM,  respectively, at day $t$. Here, $P_{tr.}$ denotes the training days which includes the 4 trading days used in the regressions in \eqref{eq:etf_two_regressions} plus an additional two trading days. The second constraint, that involves regression coefficients, aims at reaching a market-neutral portfolio and  includes a $\delta$ that is chosen to be very small; here we set $\delta = 10^{-4}$. \cite{pleroux2002} consider a similar portfolio optimization problem but with the objective of minimizing portfolio risk  subject to the portfolio returning a desired amount.

\begin{figure*}
	\centering
	\begin{subfigure}{.5\textwidth}
		\includegraphics[scale=0.32]{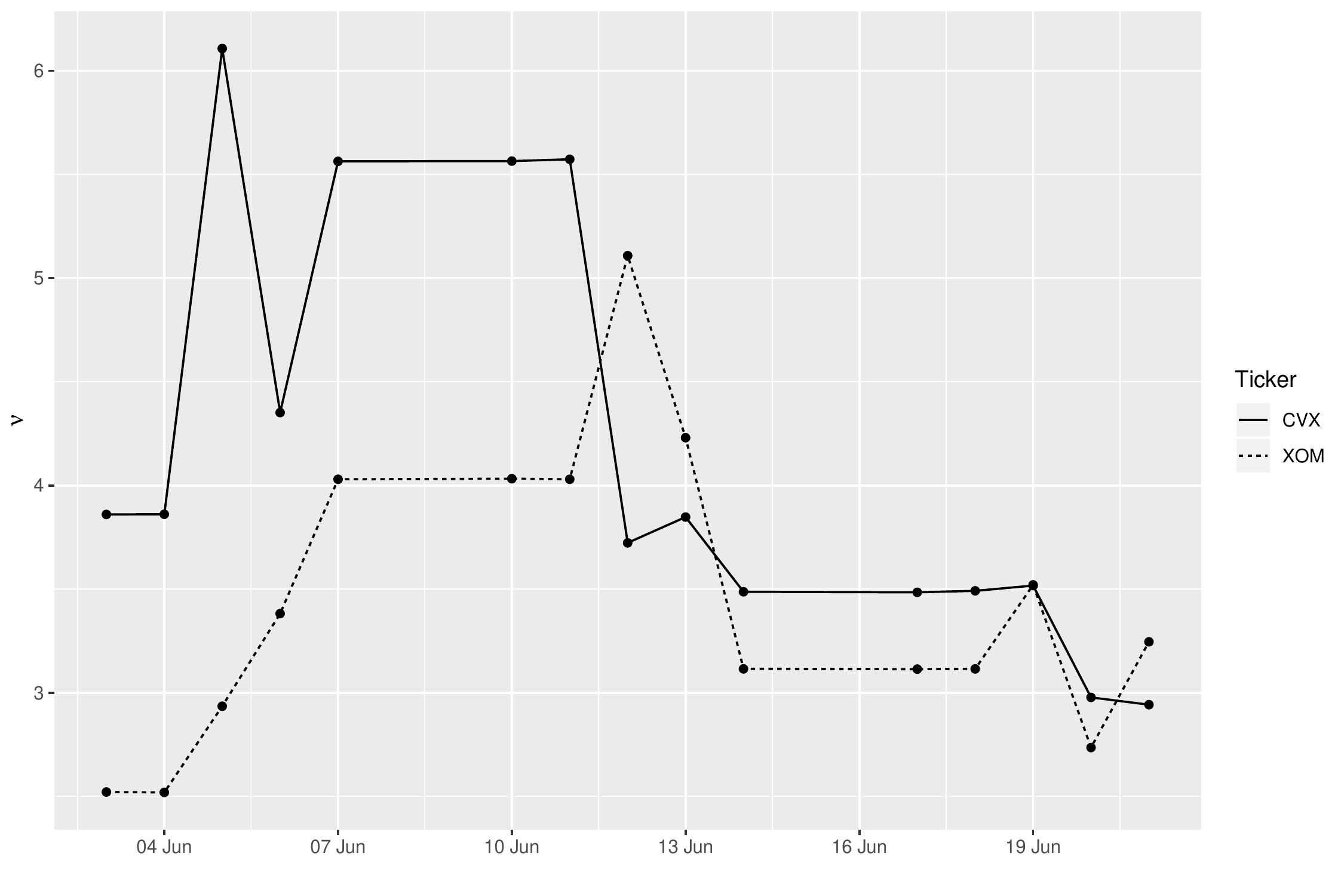}
	\end{subfigure}
	\begin{subfigure}{.45\textwidth}
		\includegraphics[scale=0.32]{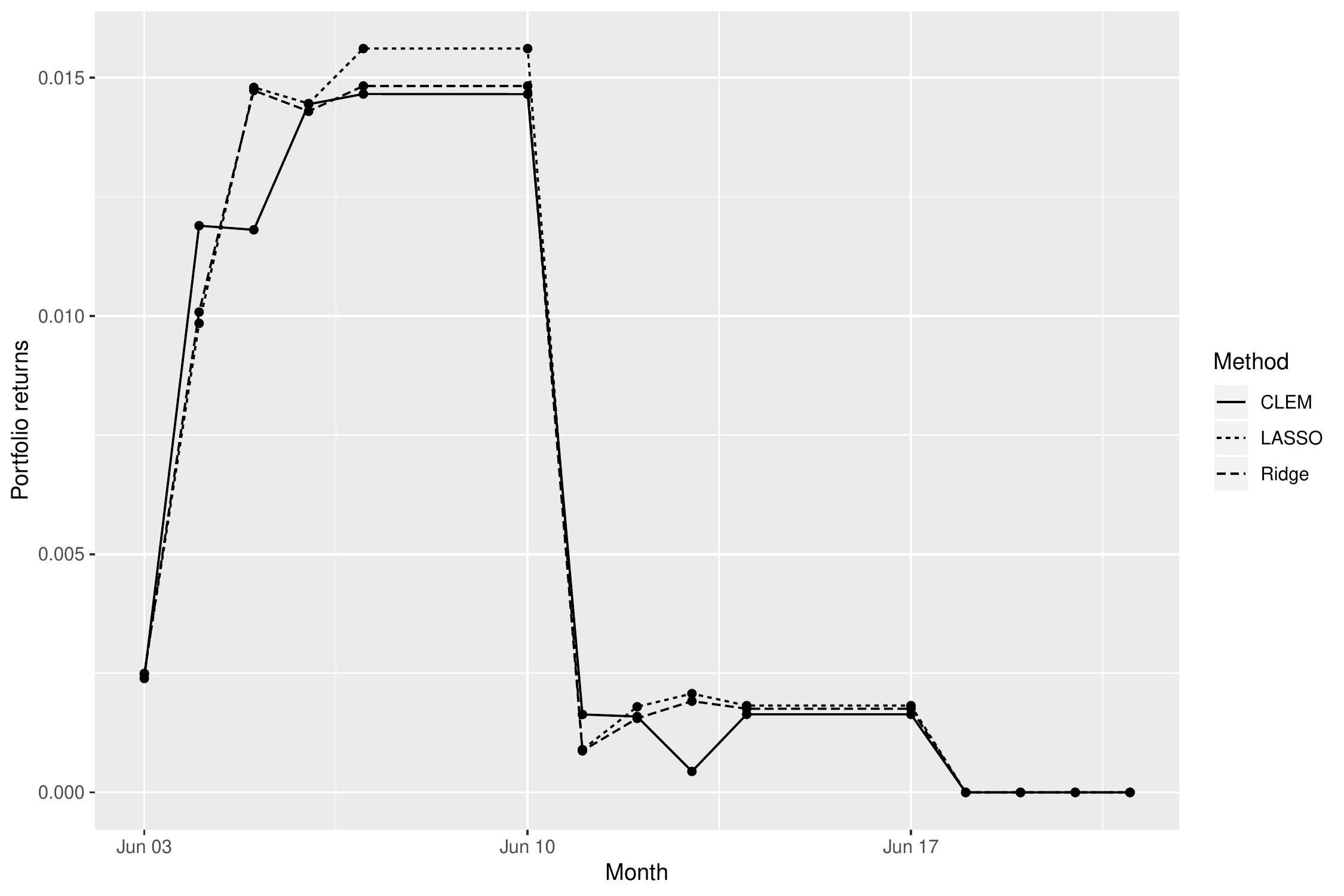}
	\end{subfigure}
	\caption{Left: Estimated $\nu$ in the 15 training periods of June 2019 for the two stocks Chevron (CVX) and Exxon (XOM). Right: Portfolio returns of the 3 competing methods during the 15 evaluation periods of June 2019.} \label{fig:etf_nu_returns}
\end{figure*}

Finally, in the right plot of Figure \ref{fig:etf_nu_returns}, we present the portfolio returns for the 3 competing methods during the 15 evaluation periods in June 2019. We notice here a comparable performance across the methods. Existing methods in the financial literature carry out either conventional linear regression or penalized regression methods to address the modeling in \eqref{eq:etf_two_regressions}. As discussed in \cite{yuan06}, sparse regression methods like LASSO are designed for selecting individual variables in a regression setup and are much less suitable in a factor regression setup with dependent factors. It is also challenging to perform statistical inference on the regression coefficients with the LASSO approach and, with the regressors and the response being time series, as in our situation, penalized regression techniques become largely inappropriate. Unlike these competitors, it must be emphasized that our method allows us to perform statistical inference on the regression coefficients. This is critical in understanding which sector ETFs are significant in impacting the individual companies' returns. Therefore, our NSVt model is more suitable for the application considered here since such an inference is fundamental to portfolio management.

\vspace{-0.5cm}

\section*{Acknowledgments and Code availability} W. Barreto-Souza would like to acknowledge the financial support from the KAUST Research Fund. The computer code for implementing our method in \textsf{R} is made available on \textsf{GitHub} at \url{https://github.com/rsundarar/student_t_SV}. Some examples from the simulations and real data analysis are included.  

\vspace{-0.7cm}

\section{Conclusions and outlook}\label{sec:conclusion}

A novel Student-t stochastic volatility model was proposed. The pairwise density function was derived in an explicit form, thereby enabling us to propose a composite likelihood inference. This approach has a low computational cost, therefore has an appealing feature over some existing stochastic volatility models where complex and time-consuming algorithms are necessary to perform inference. A composite likelihood Expectation-Maximization algorithm was developed for dealing with possible numerical issues due to the direct maximization of the composite log-likelihood function. Standard errors and confidence limits obtained via a parametric bootstrap method were also discussed. Monte Carlo simulations are performed and it was seen that both composite likelihood approaches produced satisfactory results, but with the CLEM algorithm showing a slight advantage in estimating parameters related to the latent gamma process. 

An application of our Student-$t$ process to modeling the relationship between stock price returns and sector ETF returns was explored and was further utilized to find optimal portfolios. High-frequency returns data was used which helps make portfolio decisions more frequently. The performances of our method and some competitors were assessed through portfolio returns. Our method is more suitable for the factor regression considered in this application wherein both the response and regressors are time series. Our approach allows us to perform statistical inference on the regression coefficients, as opposed to the methods in the financial literature wherein conventional regression methods or sparse regression techniques are used.

An alternative EM-algorithm for including the estimation of $\nu$ is one problem that deserves future investigation. Another topic of interest, particularly for the financial application in portfolio selection, would be the multivariate extension of the NSVt model that allows dependence among time series components. Extending our current setup based on a skew Student-$t$ distribution for dealing with skewed time series is also a potential area to explore.  

\vspace{-0.8cm}

\section*{Appendix}

Here we provide the proofs of the results stated in the paper. 

\noindent{\bf Proof of Proposition \ref{p:jointdensity}}. We can assume $\mu_t=0$ $\forall t\in\mathbb N$ without loss of generality. The non-null mean case follows directly by applying a linear transformation. By using the stochastic representation of the NSVt process, we have that
\begin{eqnarray}\label{double_integral}
	f(y_{t+j},y_t)=\int_0^\infty\int_0^\infty f(y_{t+j},y_t|z_{t+j},z_t)
	g(z_{t+j},z_t)dz_{t+j}dz_t,
\end{eqnarray}	
where $f(y_{t+j},y_t|z_{t+j},z_t)=\dfrac{\sqrt{z_{t+j}z_t}}{2\pi\sigma^2}\exp\left\{-\dfrac{y_{t+j}^2 z_{t+j}+y_t^2 z_t}{2\sigma^2}\right\}$
and $g(\cdot,\cdot)$ is the joint density function given in (\ref{eq:jointZ}). Define the real constants $a=\dfrac{(2\pi\sigma^2)^{-1}(\nu/2)^{\nu/2+1}}{\Gamma(\nu/2)(1-\rho^j)\rho^{j(\nu/4-1/2)}}$, $b=\dfrac{\nu\rho^{j/2}}{1-\rho^j}$, and $c=\left(\dfrac{2(1-\rho^j)}{\nu}\right)^{\nu/2+2}\rho^{j(\nu/4-1/2)}$.

Rearranging the terms in (\ref{double_integral}), we have that
$f(y_{t+j},y_t)=a\sum_{k=0}^\infty\dfrac{(b/2)^{2k+\nu/2-1}}{\Gamma(k+\nu/2)k!}\Omega_k(y_{t+j})\Omega_k(y_t)$,
where we have used the summation representation of the modified Bessel function of the first kind and the Monotone Convergence Theorem to interchange the signs of integral and double summation, and
\begin{eqnarray*}
	\Omega_k(y)\equiv
	\int_0^\infty z^{k+(\nu-1)/2}\exp\left\{-z\left(\dfrac{y^2}{2\sigma^2}+\dfrac{\nu}{2(1-\rho^j)}\right)\right\}dz=
	\dfrac{\Gamma(k+(\nu+1)/2)}{\left(\dfrac{y^2}{2\sigma^2}+\dfrac{\nu}{2(1-\rho^j)}\right)^{k+(\nu+1)/2}}.
\end{eqnarray*}	

It follows that
\begin{eqnarray*}
	f(y_{t+j},y_t)&=&
	ac\omega(y_{t+j},y_t)^{\frac{\nu+1}{2}}\sum_{k=0}^\infty(\rho^j\omega(y_{t+j},y_t))^k\dfrac{\Gamma(k+(\nu+1)/2)^2}{\Gamma(k+\nu/2)k!}\\
&=&ac\omega(y_{t+j},y_t)^{\frac{\nu+1}{2}}\dfrac{\Gamma\left((\nu+1)/2\right)^2}{\Gamma\left(\nu/2\right)}{_2}F_1\left(\frac{\nu+1}{2},\frac{\nu+1}{2};\frac{\nu}{2},\rho^j\omega(y_{t+j},y_t)\right),
\end{eqnarray*}
with ${_2F}_1(\cdot,\cdot;\cdot,\cdot)$ being the hypergeometric function. After using the explicit expressions for the real constants $a$ and $c$ defined at the beginning of the proof, we get the desired result.
 $\square$\\

To prove Proposition \ref{p:condvar}, we need to state two results about the Laplace transform and the conditional inverse moment of a GAR process. These results are not given in \cite{sim1990}; they are new results provided in this paper. 

\begin{lemma}\label{l:ltransf}
	Assume $\{Z_t\}_{t\in\mathbb N}\sim\mbox{GAR}(1)$. The conditional Laplace transform of $Z_{t+j}$ given $Z_t=z_t$, say $L(s)\equiv E(\exp\{-s Z_{t+j}\}|Z_t=z_t)$, for $s>0$, is given by
	\begin{eqnarray*}
		L(s)=\dfrac{\nu(z_t\rho^j)^{-(\nu/2-1)/2}}{2(1-\rho^j)}\exp\left\{-\dfrac{\nu z_t\rho^j}{2(1-\rho^j)}\right\}\sum_{k=0}^\infty\left(\dfrac{\nu\rho^{j/2}z_t^{1/2}}{1-\rho^j}\right)^{2k+\nu/2-1}\left(s+\dfrac{\nu}{2(1-\rho^j)}\right)^{-(\nu/2+k)}\bigg/k!.
	\end{eqnarray*}
\end{lemma}

\begin{proof}
	The strategy here is similar to the proof of Proposition \ref{p:jointdensity}. By using the series representation of the Bessel function and the Monotone Convergence Theorem to interchange the signs of integral and summation, the integrals involved are of the type (with $k\in\mathbb N$ denoting the index of the summation)
	\begin{eqnarray*}
		\int_0^\infty z^{\nu/2+k-1}\exp\left\{-\left(s+\dfrac{\rho^j}{2(1-\rho^j)}\right)z\right\}dz=
		\left(s+\dfrac{\nu}{2(1-\rho^j)}\right)^{-(\nu/2+k)}\Gamma(\nu/2+k).
	\end{eqnarray*}
	After some algebra, the desired result is obtained. $\square$
\end{proof}

\begin{lemma}\label{l:invmoment}
	Let $\{Z_t\}_{t\in\mathbb N}\sim\mbox{GAR}(1)$. For $\nu>2$, the conditional moment of $Z_{t+j}^{-1}$ given $Z_t=z_t$ is
	$E\left(Z_{t+j}^{-1}|Z_t=z_t\right)=\dfrac{\nu}{2(1-\rho^j)}\exp\left\{-\dfrac{\nu z_t\rho^j}{2(1-\rho^j)}\right\}\displaystyle\sum_{k=0}^\infty\dfrac{\left(\frac{\nu z_t\rho^j}{2(1-\rho^j)}\right)^k}{k!(\nu/2+k-1)}$.
\end{lemma}

\begin{proof}
	By using Eq. (2) from \cite{creetal81}, we have that $E\left(Z_{t+j}^{-1}|Z_t=z_t\right)=\int_0^\infty L(s)ds$, where $L(\cdot)$ is the Laplace transform given in Lemma \ref{l:ltransf}. We use again the Monotone Convergence Theorem to interchange the signs of integral and summation. The integrals involved here are of the type (with $k\in\mathbb N$ denoting the index of the summation)
	\begin{eqnarray*}
			\int_0^\infty \left(s+\dfrac{\nu}{2(1-\rho^j)}\right)^{-(\nu/2+k)}ds=
			\left(\dfrac{\nu}{2(1-\rho^j)}\right)^{-(\nu/2+k)+1}\bigg/(\nu/2+k-1),
	\end{eqnarray*}	
	which holds for $\nu>2$. After some algebric manipulations, we obtain the result. $\square$
\end{proof}

With Lemmas \ref{l:ltransf} and \ref{l:invmoment} at hand, we are able to prove Proposition \ref{p:condvar} as follows.\\

\noindent{\bf Proof of Proposition \ref{p:condvar}.} Let $\mathcal F_1$, $\mathcal F_2$ and $\mathcal F_3$ be $\sigma$-algebras generated by $Y_t$, $(Y_t,Z_t)$ and $(Y_t,Z_t,Z_{t+j})$, respectively. We have that $\mathcal F_1\subset\mathcal F_2\subset \mathcal F_3$. By using Theorem 5.1.6(ii) from \cite{dur2010} and our model definition, we have the following: $E(Y_{t+j}|Y_t)=E(Y_{t+j}|\mathcal F_1)=E(E(Y_{t+j}|\mathcal F_3)|\mathcal F_1)=E(Y_{t+j}|Z_{t+j})=\mu_{t+j}$. 
	
	In a similar way, $E(Y_{t+j}^2|Y_t)=E(Y_{t+j}^2|\mathcal F_1)=\mu_{t+j}^2+\sigma^2E(Z_{t+j}^{-1}|\mathcal F_1)$, where by using again the aforementioned Theorem given in \cite{dur2010}, we get $E(Z_{t+j}^{-1}|\mathcal F_1)=E(E(Z_{t+j}^{-1}|\mathcal F_2)|\mathcal F_1)=E(Z_{t+j}^{-1}|Z_t)$. Therefore, we have that $\mbox{Var}(Y_{t+j}|Y_t)=\sigma^2 E(E(Z_{t+j}^{-1}|Z_t)|Y_t)$.
	
	We now use the expression for the conditional expectation $E(Z_{t+j}^{-1}|Z_t)$ given in Lemma \ref{l:invmoment}. Therefore, by interchanging the signs of expectation and summation we obtain that the conditional variance depends on the conditional expectation of the type
	\begin{eqnarray*}
	E\left(Z_t^k\exp\left\{-\frac{\nu\rho^j Z_t}{2(1-\rho^j)}\right\}|Y_t\right)=\frac{\Gamma\left(\frac{\nu+1}{2}+k\right)}{\Gamma\left(\frac{\nu+1}{2}\right)}\left(\frac{\nu}{2}+\frac{(Y_t-\mu_t)^2}{2\sigma^2}\right)^{\frac{\nu+1}{2}}\hspace{-.3cm}\bigg/\hspace{-.2cm}	\left(\frac{\nu}{2(1-\rho^j)}+\frac{(Y_t-\mu_t)^2}{2\sigma^2}\right)^{\frac{\nu+1}{2}+k},
	\end{eqnarray*}	 
where the equality follows from the fact that $Z_t|Y_t\sim\mbox{G}\left(\frac{\nu}{2}+\frac{(Y_t-\mu_t)^2}{2\sigma^2},\frac{\nu+1}{2}\right)$. The result follows after some basic algebric manipulations. $\square$	\\

\noindent{\bf Proof of Proposition \ref{E-step}.} The joint density function of $Z_t$, $Z_{t-i}$, and $U_{ti}$ given $(Y_t,Y_{t-i})=(y_t,y_{t-i})$ is 
$f(z_t,z_{t-i},u_{ti}|y_t,y_{t-i})=
\dfrac{f(y_t|z_t)f(y_{t-i}|z_{t-i})f(z_t|u_{ti})f(z_{t-i}|u_{ti})f(u_{ti})}{f(y_t,y_{t-i})}$.
Hence, it follows that
\begin{eqnarray*}
	\zeta_{ti}=\sum_{u=0}^\infty P(U_{ti}=u) \int_0^\infty z_tf(y_t|z_t)f(z_t|u)dz_t 
	\int_0^\infty f(y_{t-i}|z_{t-i})f(z_{t-i}|u)dz_{t-i}/f(y_t,y_{t-i}).
\end{eqnarray*}

The above integral is obtained explicitly by identifying gamma kernels. Using this and noting that the series can be expressed in terms of the hypergeometric function, the result is obtained after some algebra. In a similar fashion, we get the expressions for $\zeta_{t\,t-i}$ and $\tau_{ti}$. $\square$

\vspace{-0.9cm}
\begin{singlespace}

\end{singlespace}
\end{document}